\documentclass[11pt]{article}
\usepackage{jheppub}
\usepackage{graphicx} 
\usepackage{graphbox}
\usepackage[usenames,dvipsnames]{xcolor} 
\usepackage{tikz}	
\usetikzlibrary{matrix}
\usepackage{dirtytalk}
\usepackage{dsfont}
\usepackage{colortbl}
\usepackage{braket}
\usepackage{amsthm}
\usepackage{amsmath}
\usepackage{float}
\usepackage{bbm, dsfont}
\usepackage[normalem]{ulem}
\usepackage{slashed}
\usepackage{mathtools}
\usepackage{multirow}
\usepackage{cancel}
\usepackage{cleveref}
\usepackage{orcidlink}
\usepackage{enumerate}
\usepackage{fontawesome}
\usepackage{adjustbox}
\usepackage{tcolorbox}
\usepackage{bm}
\usepackage{nicematrix}

\newcommand{\bs}[1]{\boldsymbol{#1}}
\newcommand{\rd}{\mathrm{d}}
\newcommand{\rdex}{\mathrm{d}_{\text{ext}}}

\newcommand{\ord}{\begin{cal}O\end{cal}}

\newcommand{\bx}{\bs{x}}
\newcommand{\bz}{\bs{z}}
\newcommand{\bI}{\bs{I}}
\newcommand{\bP}{\bs{P}}
\newcommand{\bA}{\bs{A}}
\newcommand{\bZ}{\bs{Z}}
\newcommand{\bSigma}{\bs{\Sigma}}
\newcommand{\bS}{\bs{S}}
\newcommand{\bF}{\bs{F}}
\newcommand{\bH}{\bs{H}}
\newcommand{\bC}{\bs{C}}
\newcommand{\bU}{\bs{U}}
\newcommand{\bR}{\bs{R}}
\newcommand{\bT}{\bs{T}}
\newcommand{\balpha}{\bs{\alpha}}
\newcommand{\bTheta}{\bs{\Theta}}
\newcommand{\bnu}{\bs{\nu}}
\newcommand{\cB}{\mathcal{B}}
\newcommand{\bIhat}{\hat{\phantom{\bI}}\!\!\!\!{\bI}}

\def\beq{\begin{equation}}
\def\eeq{\end{equation}}
\def\bsp#1\esp{\begin{split}#1\end{split}}

\newtheorem{theorem}{Theorem}

\theoremstyle{definition}
\newtheorem{example}{Example}

\newcommand{\eps}{\varepsilon}

\DeclareMathOperator{\diag}{diag}
\DeclareMathOperator{\K}{K}
\DeclareMathOperator{\E}{E}

\title{Twisted Riemann bilinear relations and Feynman integrals}
\abstract{Using the framework of 
twisted cohomology, we study twisted Riemann bilinear relations (TRBRs) satisfied by multi-loop Feynman integrals and their cuts in dimensional regularisation. After showing how to associate to a given family of Feynman integrals a period matrix whose entries are cuts, we investigate the TRBRs satisfied by this period matrix, its dual and the intersection matrices for twisted cycles and co-cycles. For maximal cuts, the non-relative framework is applicable, and the period matrix and its dual are related in a simple manner. We then find that the TRBRs give rise to quadratic relations that generalise quadratic relations that have previously appeared in the literature. However, we find that the TRBRs do not allow us to obtain quadratic relations for non-maximal cuts or completely uncut Feynman integrals. This can be traced back to the fact that the TRBRs are not quadratic in the period matrix, but separately linear in the period matrix and its dual, and the two are not simply related in the case of a relative cohomology theory, which is required for non-maximal cuts. 

}
\author[a]{Claude Duhr,}
\emailAdd{cduhr@uni-bonn.de}
\affiliation[a]{Bethe Center for Theoretical Physics, Universität Bonn, D-53115, Germany
}

\author[a
]{Franziska Porkert,}
\emailAdd{fporkert@uni-bonn.de}

\author[a]{Cathrin Semper,}
\emailAdd{csemper@uni-bonn.de}

\author[a
]{Sven F. Stawinski,}
\emailAdd{sstawins@uni-bonn.de}

\begin{document}

\begin{flushright}
BONN-TH-2024-10
\end{flushright}

\definecolor{lightorange}{rgb}{1,0.6875,0.5}
\definecolor{lightblue}{rgb}{0.68, 0.85, 0.9}

\maketitle


\newpage
\section{Introduction}
\label{sec.intro}

Scalar multi-loop Feynman integrals are a cornerstone for the computation of observables in Quantum Field Theory and are needed to make precise predictions for collider or gravitational wave experiments. Understanding their mathematical properties is therefore crucial, leading to significant recent efforts in this area. Feynman integrals often diverge, and need to be regularised. The most commonly used regularisation scheme is dimensional regularisation~\cite{THOOFT1972189}, where the integrals are analytically continued to an arbitrary space-time dimension $D=d-2\eps$, with $d$ an integer. The integrals are interpreted as meromorphic functions of the space-time dimension, or equivalently of $\eps$, and one is then typically interested in the Laurent-expansion of the integrals around $\eps=0$.

It is well known that dimensionally-regulated Feynman integrals are not  independent, but they satisfy relations. These relations are invaluable for applications, as they reduce the computational effort: Knowing a minimal set of Feynman integrals reduces the number of Feynman integrals that need to be evaluated and avoids hidden zeroes. Linear relations between Feynman integrals are well studied~\cite{Tkachov:1981wb,Chetyrkin:1981qh,Tarasov:1996br,Lee:2009dh}. In particular, we have efficient algorithms to solve linear relations among Feynman integrals and to identify a basis for them, see, e.g., refs.~\cite{Laporta:2000dsw,Anastasiou:2004vj,Smirnov:2008iw,vonManteuffel:2012np,Lee:2012cn,Lee:2013mka,Peraro:2019svx,vonManteuffel:2014ixa,Wu:2023upw,Georgoudis:2016wff,Klappert:2020nbg,Maierhofer:2017gsa}. 
Recently, it was realised that the appropriate mathematical setup to study dimensionally-regulated Feynman integrals is twisted cohomology~\cite{aomoto_theory_2011,Mizera:2017rqa,Mastrolia:2018uzb}. Twisted cohomology can be described, loosely speaking, as a mathematical framework to study integrals that depend on multi-valued integrands. Linear relations among Feynman integrals naturally arise by identifying integrands that only differ by a total (covariant) derivative. This has led to a completely new view on linear relations, including new methods to find bases for Feynman integrals~\cite{Frellesvig:2019kgj,Frellesvig:2020qot,Chestnov:2022xsy,Brunello:2023rpq,Fontana:2023amt} and to solve differential equations for Feynman integrals~\cite{Chen:2020uyk,Chen:2022lzr,Chen:2023kgw,Jiang:2023jmk}.

The main focus of this paper is to take steps towards understanding if there are also non-linear relations between Feynman integrals. 
While linear relations between (cut) Feynman integrals are very well understood, this is not the case for non-linear relations. 
%
The first appearance of quadratic relations among Feynman integrals was in the paper by Broadhurst and Roberts~\cite{Broadhurst:2018tey} (see also refs.~\cite{Zhou:2017jnm,Zhou:2017vhw,Zhou:2018tva,Zhou:2018wyp,Fresan:2020anx}), 
which presents quadratic relations involving equal-mass banana integrals evaluated at $p^2=m^2$.
%
Quadratic relations involving the maximal cuts of Feynman integrals in general kinematics were presented in refs.~\cite{Duhr:2022dxb} and~\cite{Lee:2018jsw}. The focus of ref.~\cite{Duhr:2022dxb} are maximal cuts in integer dimensions that evaluate to (quasi-)periods of Calabi-Yau (CY) varieties (cf.,~e.g.,~refs.~\cite{Bloch:2014qca,Bloch:2016izu,Bourjaily:2018ycu,Bourjaily:2018yfy,Bourjaily:2019hmc,Klemm:2019dbm,Bonisch:2020qmm,Bonisch:2021yfw,Forum:2022lpz,Lairez:2022zkj,delaCruz:2024xit}), and it is well known that (quasi-)periods of CY varieties satisfy a set of quadratic relations~\cite{MR717607}.
%
While the previous examples of quadratic relations only hold for integrals evaluated at $\eps=0$, quadratic relations valid in dimensional regularisation with $\eps\neq0$ were presented in ref.~\cite{Lee:2018jsw} (and further elaborated on in ref.~\cite{Lee:2019lsr}) for maximal cuts of Feynman integrals depending on a single dimensionless kinematic variable $x$. The results of ref.~\cite{Lee:2018jsw} follow from a set of conjectural properties satisfied by the system of differential equations which these maximal cuts obey. 

From the explicit examples of quadratic relations in the literature, one can observe that they all take the schematic matrix form
 \beq\label{eq:schematic_quad_rel}
 \bP(\bx,-\eps)^T\bR(\bx,\eps) \bP(\bx,\eps) = \widetilde{\bH}(\eps)\,,
 \eeq
 where $\bP(\bx,\eps)$, $\bR(\bx,\eps)$ and $\widetilde\bH(\eps)$ are matrices, and $\widetilde\bH(\eps)$ is independent of the kinematic variables.
 The goal of this paper is to initiate a general analysis of quadratic relations involving Feynman integrals and their cuts, and to understand in how far it is possible to generalise the aforementioned quadratic relations among maximal cuts to other Feynman integrals and/or to non-maximally cut integrals. In fact, the framework of twisted cohomology naturally contains a set of relations that are bilinear in transcendental integrals, namely the so-called \emph{twisted Riemann bilinear relations} (TRBRs)~\cite{Cho_Matsumoto_1995}, which are a generalisation of the well-known Riemann bilinear relations satisfied by the periods of Riemann surfaces. Due to the generality of these TRBRs, one may expect there to be quadratic relations involving Feynman integrals quite generally, though explicit examples have not been worked out. The existence of such quadratic relations may have far reaching consequences, as they may for example be used to express complicated Feynman integrals in terms of simpler ones. One of the main results of this paper is that the TRBRs indeed allow one to define quadratic relations among maximal cut integrals. However, in contrast to the folkloristic expectation, they do not lead to quadratic relations among non-maximally cut integrals or completely uncut integrals. 
 
 This paper is organised as follows: In section~\ref{sec:FIcuts} we define Feynman integrals in dimensional regularisation and recall their main properties. In section~\ref{sec:twisted} we give a brief review of (relative) twisted cohomology and the TRBRs.
 In section~\ref{sec:RelationsRiemann} we investigate TRBRs for cut Feynman integrals and we explain why they do not lead to quadratic relations among non-maximal cuts. Finally, in section~\ref{sec:MaxCutRelations} we discuss TRBRs for maximal cuts and we show how known quadratic relations can be derived from them, and in section~\ref{sec:examples} we present explicit examples of quadratic relations for one- and two-loop integrals. In section~\ref{sec:concl} we draw our conclusion. We also include several appendices with proofs and other details omitted throughout the main text.
 
%
%

\section{Feynman integrals and their cuts}
\label{sec:FIcuts}

\subsection{Feynman integrals and the Baikov representation}

We consider $L$-loop scalar Feynman integrals of the form
\begin{align}
\label{eq_sec2.1}
I_{\bs{\nu}}^D \left(\{p_i\cdot p_j\},\{m_i^2\} \right)= e^{L\gamma_E \varepsilon}\int \left(\prod_{j=1}^L \frac{d^D \ell_j}{i\pi^\frac{D}{2}}\right) \frac{1}{D_1^{\nu_1}\dots D_m^{\nu_m}}\, , 
\end{align}
where $\bs{\nu} = (\nu_1,\ldots,\nu_m)$ and $\gamma_E=-\Gamma'(1)$ is the Euler-Mascheroni constant. The momenta flowing through the $m$ propagators are linear combinations of the $E$ linearly independent external momenta $p_j$ and the $L$ loop momenta $\ell_j$, and we denote the squared mass of the $i^{\textrm{th}}$ propagator by $m_i^2$:  
\begin{align}
    D_i =\left(\sum_{j=1}^L a_{ij}\ell_j + \sum_{j=1}^{E-1} b_{ij}p_j\right)^2 - m_i^2 \, ,\qquad a_{ij},b_{ij}\in \{0,\pm1\}\,.
\end{align}
Throughout this paper, we work in dimensional regularisation in $D=d-2\varepsilon$ dimensions, with $d$ a positive integer. Unless specified otherwise, the exponents $\nu_i$ of the propagators will be integers (both positive and negative).

While eq.~\eqref{eq_sec2.1} is natural from a physics perspective and has a direct and clear connection to Feynman graphs, it is often not suited for evaluating the integrals or for studying their properties. 
In the remainder of this paper, we will work with the  \textit{Baikov representation} of a Feynman integral~\cite{BAIKOV1997347} (see also ref.~\cite{LEE2010474}). The Baikov representation is obtained by changing variables in eq.~\eqref{eq_sec2.1} from the components of the loop momenta to the propagators $z_i:=D_i$. This leads to the integral representation:
\begin{align}
\label{eq_Baikov}
I_{\bs{\nu}}^D \left(\{p_i\cdot p_j\},\{m_i^2\} \right)&= \frac{e^{L\gamma_E\varepsilon} \left[\det G(p_1,\ldots,p_E)\right]^\frac{-D+E+1}{2}}{\pi^{\frac{1}{2}(N-L)} \left[\det C\right] \prod_{j=1}^L \Gamma\left(\frac{D-E+1-j}{2}\right)}\,\hat{I}_{\bs{\nu}}^D \left(\{p_i\cdot p_j\},\{m_i^2\} \right)\,,
\end{align}
with
\begin{align}
\label{eq_Baikov2}
\hat{I}_{\bs{\nu}}^D \left(\{p_i\cdot p_j\},\{m_i^2\} \right)&=\int_{\mathcal{C}} \rd^{N} z \left[\mathcal{B}(\bs{z})\right]^{\frac{D-L-E-1}{2}} \prod_{s=1}^{N} z_s^{-\nu_s}\, . 
\end{align}
The Gram determinants are defined as
\begin{align}
    \det G(q_1,\dots, q_n)=\det \left(-q_i\cdot q_j\right)\,,
\end{align}
and the \textit{Baikov polynomial} is
\begin{align}
    \mathcal{B}(\bs{z}) =\det G(\ell_1,\ldots,\ell_L,p_1,\ldots,p_E) \,,
\end{align}
with $\bs{z}=(z_1,\ldots,z_{N})$.\footnote{We follow the convention that vectors and matrices are denoted by boldfaced letters.} The determinant $\det C$ is the Jacobian of the transformation from $q_i\cdot q_j$, the independent scalar products of internal and external momenta, to the propagators $z_i=D_i$, and it is independent of the $z_i$. The integration contour $\mathcal{C}$ is given by $\mathcal{C}=\mathcal{C}_1\cap\dots \cap \mathcal{C}_L$, with
\begin{align}
\mathcal{C}_j=\left\{\bs{z}\in\mathbb{R}^{N}: \frac{\det G(\ell_j,\ell_{j+1},\dots,\ell_L,p_1,\dots, p_E)}{\det G(\ell_{j+1},\dots, \ell_L,p_1,\dots, p_E)} \geq 0  \right\}\, . 
\end{align}
The number of variables $z_i$ is the number $N$ of linearly independent scalar products involving the loop momenta:
\begin{align}
    N= \frac{1}{2}L (L+1)+EL\, .  
\end{align}

Let us make some comments about the Baikov representation. 
First, if the number of propagators is smaller than $N$,  we add propagators and set their power $\nu_i$ to zero. 
 Second, the integrand of eq.~\eqref{eq_Baikov} is multi-valued, with the multi-valuedness originating from the non-integer exponent of the Baikov polynomial in dimensional regularization. Finally, for multi-loop cases ($L>1$), it is often beneficial to introduce a Baikov parametrization for each loop separately. This so-called \emph{loop-by-loop approach}~\cite{Frellesvig:2017aai} will typically lead to integrals of lower dimension, but with a product of Baikov polynomials in the integrand, each with a different (non-integer) exponent. More precisely, in the loop-by-loop approach, the Feynman integral is proportional to an integral of the form
 \beq\label{Baikov_LL}
 \int_{\mathcal{C}'} \rd^{N'}\! z\,
 \mathcal{B}_1(\bs{z})^{\mu_1}\dots \mathcal{B}_K(\bs{z})^{\mu_K}\prod_{s=1}^{N'} z_s^{-\nu_s}\,,
 \eeq
where $N'\le N$ and each exponent $\mu_i$ is linear in the dimensional regulator $\eps$. The explicit form of the Baikov polynomials $\mathcal{B}_i$, the exponents $\mu_i$, the contour $\mathcal{C}'$ and the proportionality factor depend on the loops that have been integrated out, and the order of these integrations. 


We are not only interested in Feynman integrals as defined in eq.~\eqref{eq_sec2.1}, but also in their \emph{cuts}. Different notions of cuts have appeared in the literature (see, e.g., ref.~\cite{Britto:2024mna} for a recent review), but they all have in common that they involve a notion of putting a subset of propagators on their mass shell. Our definition here is most transparent in the Baikov representation: A cut of a Feynman integral is obtained by taking the residue at the origin in a subset of variables $z_i$  of the differential form defining the Baikov representation in eq.~\eqref{eq_Baikov} (or its loop-by-loop counterparts), times $(2 \pi i)^{\lfloor \frac{n_c}{2}\rfloor}$, where $n_c$ denotes the number of cut propagators. Note that our definition implies that a cut vanishes whenever we take a residue in a variable $z_i$ with $\nu_i\le 0$, because in that case the integrand in eq.~\eqref{eq_Baikov} is regular at $z_i=0$. Finally, we note that a special role is played by so-called \emph{maximal cuts}, which correspond to cuts obtained by taking the residue in all variables $z_i$ with $\nu_i>0$.


\subsection{Linear relations and differential equations for (cut) Feynman integrals}

Before we discuss quadratic relations, we briefly review what is known about linear relations. There are two classes of linear relations among Feynman integrals, corresponding to whether the dimension $D$ is held fixed or not. 

Feynman integrals in the same dimension $D$ and with the same set of propagators $D_i$, but for different values of the exponents $\nu_i$, can be collected into an {integral family}, and there are linear relations among different members of the same family. They arise from \emph{integration-by-parts} (IBP) relations~\cite{Tkachov:1981wb,Chetyrkin:1981qh} (and symmetries of the underlying Feynman graph). IBP relations can be succinctly captured by the identity
\beq\label{eq:IBP_tot_diff}
\int \rd^D \ell_i\,\frac{\partial}{\partial \ell_i^\mu}\left(\frac{v^{\mu}}{D_1^{\nu_1}\dots D_m^{\nu_m}} \right)=0\,,
\eeq
where $v^\mu$ can be either an internal or external momentum. When the derivative acts on the propagators, it produces propagators with shifted exponents (plus numerator factors that can again be expressed in terms of inverse propagators). In this way the left-hand side of eq.~\eqref{eq:IBP_tot_diff} can be expressed as a linear combination of integrals in the same dimension $D$, but with shifted exponents $\nu_i$. One can solve the IBP relations and express all members of a family in terms of some basis integrals, often called \emph{master integrals} in the literature. There is of course a considerable freedom in how to choose the basis of master integrals. We assume from now on that such a basis has been fixed, and we collect the master integrals into the vector ${\bIhat}(\bx,\eps)$, where $\bx$ collectively denotes the (dimensionless) kinematic variables on which the integrals depend.\footnote{In the following, hatted quantities indicate Feynman integrals rescaled as in eq.~\eqref{eq_Baikov2}.}

The IBP relations connect integrals with different values of the exponents $\nu_i$, but in the same dimension.
The Baikov representation in eq.~\eqref{eq_Baikov2} makes it manifest that there is no fundamental difference between the dimension $D$ and the exponents $\nu_i$: they both appear as exponents of the different factors in the integrand. Correspondingly, there are also linear relations between integrals in $D$ and $D\pm2$ dimensions. Indeed, from eq.~\eqref{eq_Baikov2} we can immediately see that
\begin{align}
\hat{I}_{\bs{\nu}}^{D+2} \left(\{p_i\cdot p_j\},\{m_i^2\} \right)&=\int_{\mathcal{C}} \rd^{N} z \left[\mathcal{B}(\bs{z})\right]^{\frac{D-L-E}{2}}\,\mathcal{B}(\bs{z})\, \prod_{s=1}^{N} z_s^{-\nu_s}\, ,
\end{align}
and since $\cB(\bz)$ is a polynomial in $\bz$, we can write the right-hand side as a linear combination of integrals with shifted exponents but in $D$ dimensions. These so-called \emph{dimension-shift relations} were first derived in ref.~\cite{Tarasov:1996br} (see also ref.~\cite{Lee:2009dh}). The dimension-shift relations between master integrals in $D+2$ and $D$ dimensions lead to linear relations among the master integrals:
\begin{align}\label{eq:R_dim_shift}
    \bIhat(\bs{x},\varepsilon-1) = \bs{R}(\bs{x},\varepsilon) \bIhat(\bs{x},\varepsilon)\, . 
\end{align}
In the rest of this paper, we will refer to the matrix $\bs{R}(\bs{x},\varepsilon)$ as the \textit{dimension-shift matrix}.

 Let us finally make some comments about how these relations extend to cut integrals. Since we pass from uncut to cut integrals by taking a residue at $z_i=0$ for some values of $i$, it follows that all linear relations (both IBP and dimension-shift relations) for uncut integrals carry over to cuts, but we need to put to zero all terms with integrals that involve negative powers. This simple rule, which allows one to recover linear relations for cut integrals from their uncut analogues, is known as  \emph{reverse-unitarity} in the literature~\cite{Anastasiou:2002yz,Anastasiou:2003yy}. We do therefore not discuss linear relations among cut integrals explicitly.

As a consequence of the IBP relations, the master integrals satisfy a system of linear differential equations of the form~\cite{Kotikov:1990kg,Kotikov:1991hm,Kotikov:1991pm,Gehrmann:1999as,Henn:2013pwa},
\begin{align}\label{eq:DEQ_generic}
    \rd_{\text{ext}} \bIhat(\bs{x},\varepsilon) =\bs{\Omega} (\bs{x},\varepsilon) \bIhat(\bs{x},\varepsilon)\, , 
\end{align}
where $\rd_{\text{ext}}=\sum_i\rd x_i\,\partial_{x_i}$ is the exterior derivative with respect to the external (kinematic) parameters, and $\bs{\Omega} (\bs{x},\varepsilon)$ is a matrix whose entries are rational one-forms. Since the linear relations satisfied by cut integrals can be obtained from their uncut analogues, we can in the same way obtain a system of differential equations satisfied by cut integrals.

\section{(Relative) twisted (co-)homology}
\label{sec:twisted}

It was recently shown that the right mathematical setup to define and study dimensionally-regulated Feynman integrals and the linear relations they satisfy is twisted cohomology~\cite{Mastrolia:2018uzb}.
Twisted cohomology can be summarised, loosely speaking, as the study of integrals of the form
\begin{align}
\label{gen_twisted_integral}
    \int_{\Gamma} \Phi \varphi\, , 
\end{align}
where $\Phi$ is a multi-valued function and  $\varphi$ is a rational $n$-form on $X=\mathbb{C}^n-\Sigma$ with $\Sigma$ a union of hypersurfaces to be defined below. The multi-valued function $\Phi$ is called the \emph{twist}, and we take it to be of the form
\begin{align}
\label{twist32}
     \Phi= \prod_{i=0}^{r} L_i(\bs{z})^{\alpha_i} \,,
\end{align}
where  the $L_i(\bs{z})$ are polynomials in the variables $\bs{z}\in X$ and $\alpha_i\in\mathbb{C}$ in the most general case.
For the integral in eq.~\eqref{gen_twisted_integral} to be well-defined, a choice of branch must be made for $\Phi$. 

If $\Phi \varphi$ has branch points, but no poles or zeroes (i.e. $\alpha_i\notin\mathbb{Z}$ for all $i$ and $\sum_{i=0}^r\alpha_i\notin\mathbb{Z}$ ), the natural framework for integrals as in eq.~\eqref{gen_twisted_integral} is \textit{(non-relative) twisted (co-)homology}. We review relevant aspects of twisted (co)-homology in section~\ref{subsec:twisted}. For Feynman integrals, we typically encounter the situation where $\Phi\varphi$ also has poles. In that case, \textit{relative twisted (co-)homology} is the natural framework, and we will review it in section~\ref{subsec:reltwist}. 
 We focus here on the basics necessary for understanding quadratic relations in later sections. For a more in-depth review of twisted (co-)homology and its application we refer to the literature \cite{yoshida_hypergeometric_1997,aomoto_theory_2011,Mizera:2017rqa,Mizera:2019gea,Mizera:2019ose,matsumoto_relative_2019-1,Caron-Huot:2021iev,Caron-Huot:2021xqj,Giroux:2022wav,Crisanti:2024onv,Gasparotto:2023roh,Fontana:2023amt,Bhardwaj:2023vvm,De:2023xue,Britto:2021prf,Duhr:2023bku,Brunello:2023rpq}.

\subsection{Twisted (co-)homology}
\label{subsec:twisted}
We start by restricting the exponents in the twist in eq.~\eqref{twist32} by the condition
\begin{align}
\label{restrict}
\alpha_i\notin \mathbb{Z}\text{~~~~~and~~~~~} \sum_{i=0}^{r} \alpha_i \notin \mathbb{Z}\, .
\end{align}
In this case, $\Phi\varphi$ is a multi-valued, holomorphic form on the space  $X$, with $\Sigma$ the union of the zero loci of the $L_i(\bs{z})$. We call the different components of $\Sigma$ the \textit{regulated boundaries}. The twist defines a connection that accounts for the multi-valuedness of the integrals
\begin{align}
\label{connect}
\nabla = \rd +\omega \wedge \cdot \text{~~~~with~~~~} \omega = \frac{\rd \Phi}{\Phi} =\rd\!\log \Phi\, \, . 
\end{align}Working modulo exact forms -- i.e., identifying $\varphi$ with $\varphi +\nabla \tilde{\varphi}$ --means working modulo forms that vanish upon integration.  
Since we are interested in integrals, we
only consider elements of the \textit{twisted cohomology group} 
\beq
\label{twistedcoh}
    H_{\text{dR}}^k (X, \nabla)=C^k(X,\nabla)/B^k(X,\nabla)\,, 
    \eeq
    with
    \beq\bsp
    \label{CBeq}
    C^k(X,\nabla)&\,=\{ k-\text{forms }\varphi \text{ on } X\, :\, \nabla \varphi = 0 \}\,,\\
    B^k(X,\nabla)&\,= \{ k-\text{forms }\nabla\tilde{\varphi}\, :\, \tilde{\varphi} \text{ a $k-1$-form}\} \, . 
\esp\eeq
All twisted cohomology groups are finite-dimensional, and only the middle cohomology group $H_{\text{dR}}^{n}(X,\nabla)$ is nonzero~\cite{aomoto_theory_2011}. We therefore only consider $k=n$ from here on. The elements of $H_{\text{dR}}^{n}(X,\nabla)$ are called \textit{twisted co-cycles}. The relevant integration contours are \textit{twisted cycles} (also called \emph{loaded cycles}) from 
\begin{align}
\label{homgroup}
    C_n(X,\check{\mathcal{L}}) =  \{n-\text{cycles }\gamma \otimes \Phi|_{\gamma}\,  :\, \partial (\gamma\otimes \Phi|_{\gamma} )=0\} \, . 
\end{align}
Here $\gamma \otimes \Phi|_\gamma$ denotes an $n$-cycle $\gamma$ on $X$ together with a branch of $\Phi$, specified by the locally constant sheaf $\check{\mathcal{L}}$ generated by $\Phi$. In practice, this means that every twisted cycle comes with a local choice of branch for $\Phi$. 
As usually no confusion arises, we will often identify $\gamma\otimes \Phi$ with $\gamma$. The operation $\partial (\gamma\otimes\Phi|_{\gamma})$ restricts the contour with its branch of $\Phi$ to its boundary. We denote the space of all boundaries by $B_n(X,\check{\mathcal{L}})= \{n-\text{cycles }\partial \left(\gamma\otimes \Phi|_{\gamma}\right)\}$. We only want to consider closed contours on $X$ modulo boundaries. We therefore consider elements of the \textit{twisted homology group}:
\begin{align}
\label{twistedho}
H_n(X, \check{\mathcal{L}}) =C_n(X,\check{\mathcal{L}})/ B_n(X,\check{\mathcal{L}}) \, . 
\end{align}
For the cases considered here, the cycles are (regularised) chambers between regulated boundaries. The details of the regularisation are explained in refs.~\cite{yoshida_hypergeometric_1997,kita_intersection_1994-2,aomoto_theory_2011, Bhardwaj:2023vvm, Duhr:2023bku}.

We can pair twisted cycles and co-cycles to obtain integrals such as those in eq.~\eqref{gen_twisted_integral}. This defines the \textit{period pairing}: 
\begin{align}
 \langle \gamma |\varphi]= \int_{\gamma} \Phi \varphi\, . 
\end{align}
Pairing basis elements $\gamma_i$ of $H_n(X,\check{\mathcal{L}})$ and $\varphi_j$ of $H_{\text{dR}}^n(X,\nabla)$, we obtain the period matrix $\bs{P}$ with entries
\begin{align}
 P_{ij}= \langle\varphi_i| \gamma_j]\  =\int_{\gamma_j}\Phi\varphi_i\, . 
\end{align}
For the integrations in the period matrix the details of the regularisation can usually be ignored and the twisted cycles can be treated as non-regularised chambers between boundaries (usual intervals in one dimensions), see ref.~\cite{yoshida_hypergeometric_1997}. 

By replacing the twist with its inverse in eq.~\eqref{connect}, one defines the dual connection $\check{\nabla}=\rd -\omega\wedge \cdot $, and one obtains, similar to eqs.~\eqref{twistedcoh} and \eqref{twistedho}, the corresponding dual twisted (co-)homology groups:
\beq\bsp
H_{\text{dR},c}^n \left(X, \check{\nabla}\right) &= \{ \text{compactly supported } n-\text{forms } \check{\varphi} \, :\, \check{\nabla} \check{\varphi} =0\}/ \{\text{exact forms}\}\,,\\
H_n^{\text{lf}}(X,\mathcal{L}) &=  \{\check{\gamma} \otimes \Phi^{-1}|_{\check\gamma} \text{ locally finite} \, :\,  \partial \check{\gamma}  =0\} / \{ \text{boundaries }\partial\tilde{\gamma }\} \, .
\esp\eeq
The restriction to compactly supported forms for the dual cohomology group is necessary for the intersection pairing between twisted co-cycles to be well-defined,
\begin{align}
    \langle \varphi_A|\check\varphi_{B}\rangle&=\int_X \varphi_A\wedge \check\varphi_{B}\,.
\end{align}
For a choice of basis of the twisted cohomology group and its dual, we define the intersection matrix $\bs{C}$ with entries
\begin{align}
    C_{ij} = \frac{1}{(2\pi i)^n} \langle \varphi_i |\check{\varphi}_j\rangle \, . 
\end{align}
The intersection pairing $
[\check{\gamma}_i|\gamma_j]$ for the twisted homology group and its dual counts the (topological) intersections of the two contours, taking into account their orientations as well as the branch choices for $\Phi$ and $\Phi^{-1}$ loaded onto them. More details on their computation can be found in refs.~\cite{yoshida_hypergeometric_1997,Duhr:2023bku,kita_intersection_1994-2}. For a choice of basis of the twisted homology group and its dual, the intersection matrix $\bs{H}$ is the matrix with entries 
\begin{align}
    H_{ij}=[\check{\gamma}_j|\gamma_i]\, . 
\end{align}
Since the (non-dual) twisted cycles are already regularised, the dual cycles are only required to be locally finite for the intersection pairings of $\bs{H}$ to be well-defined, i.e., they can be chosen as line segments (or chambers in higher dimension) without the regularisation required for the twisted cycles. Pairing the basis elements of both dual  groups, one obtains the dual period matrix $\check{\bs{P}}$ with entries
\begin{align}
\label{}
    \check{P}_{ij} = [\check{\gamma}_j|\check{\varphi}_i\rangle = \int_{\check{\gamma}_j} \Phi^{-1} \check{\varphi}_i\, .  
\end{align}
Let us make two comments. First, all four pairings (the intersection pairings in homology and cohomology and the (dual) period pairing) are non-degenerate, and so the matrices $\bP$, $\bs{\check{P}}$, $\bC$ and $\bH$ all have full rank. Second, in applications the differential forms, and thus the periods, depend on some parameters $\bx$ (the external kinematic data in the case of Feynman integrals). The period matrix and its dual are the fundamental solution matrices of the differential equations,
\beq\bsp\label{eq:Gauss-Manin}
\rdex \bP(\bx,{\balpha}) &\,= \bs\Omega(\bx,\balpha)\bP(\bx,\balpha)\,,\\
\rdex \bs{\check{P}}(\bx,{\balpha}) &\,= \bs{\check{\Omega}}(\bx,\balpha)\bs{\check{P}}(\bx,\balpha)\,.
\esp\eeq
The intersection matrix $\bC$ is the unique rational solution of the equation~\cite{Chestnov:2022alh,Chestnov:2022okt,Munch:2022ouq}
\beq\label{eq:dC}
\rdex\bC(\bx,\balpha) = \bs\Omega(\bx,\balpha)\bC(\bx,\balpha) + \bC(\bx,\balpha)\bs{\check{\Omega}}(\bx,\balpha)^T\,.
\eeq
Note that eqs.~\eqref{eq:Gauss-Manin} and~\eqref{eq:dC} can only hold if the differential $\rdex$ with respect to the external parameters $\bx$ does not receive contributions from the integration cycles, but all the information on the differentials is encoded into the twisted co-cycles. In the following, we always assume that this condition is satisfied.

We now stress a property that holds whenever the exponents $\alpha_i$ satisfy eq.~\eqref{restrict}. In this case, the dual of the twisted cohomology group is  given by the twisted cohomology group with compact support. For a given $n$-form $\varphi$ one can compute a compactly supported version in the same cohomology class with the methods described  in refs.~\cite{Mizera:2017rqa,Mastrolia:2018uzb}. We denote this compactly supported version by $\left[\varphi\right]_c$. We can then choose as a basis of the dual twisted cohomology group the compactly supported version of the basis of the twisted cohomlogy group:
\beq\label{eq:check_to_c}
\check{\varphi}_i = [\varphi_i]_c\,.
\eeq
Similarly, a basis cycle can be chosen to be the regularised version of the dual basis cycle~\cite{Cho_Matsumoto_1995}:
\beq\label{eq:check_to_h}
\gamma_i=[\check{\gamma}_i ]_\text{reg}\,.
\eeq
Upon the integral pairing we obtain 
\beq
\int_{\check{\gamma}_j }\Phi^{-1}\check{\varphi}_i=\int_{\check{\gamma}_j }\Phi^{-1}[\varphi_i]_c=\int_{[\check{\gamma}_j]_\text{reg} }\Phi^{-1}\varphi_i=\int_{\gamma_j}\Phi^{-1}\varphi_i\,.
\eeq
Thus, the difference between the period matrix and its dual only lies in the choice of the twist $\Phi^{-1}$ instead of $\Phi$. More explicitly, we then obtain:
\begin{align} \label{eq:checkP_to_minus}
\check{P}_{ij}=\int_{\check{\gamma}_j} \Phi^{-1} \check{\varphi}_i = \int_{{\gamma}_j} {\varphi}_i\prod_{k=0}^r L_k(\bs{z})^{-\alpha_k}= P_{ij}|_{\alpha_k\rightarrow -\alpha_k}\, . 
\end{align}
In other words, in the case where the exponents in the twist satisfy condition~\eqref{restrict}, we can choose a basis of dual (co-)cycles such that the dual twisted period matrix $\bs{\check{P}}$ agrees with the twisted period matrix $\bP$, up to changing the signs of the exponents $\alpha_i$.

We have seen that, for a given choice of basis, the four pairings between the different twisted (co-)homology groups and their duals give rise to the four matrices $\bP$, $\bs{\check{P}}$, $\bC$ and $\bH$.
There are also \textit{completeness relations} 
\beq\bsp
\label{eq:completeness}
(2 \pi i)^{-n} \ket{\check{\varphi}_i} (\bC^{-1})_{ij} \bra{\varphi_j} &\, = \mathds{1}\,,\\
|\gamma_i](\bH^{-1})_{ji}[\check{\gamma}_j|&\, = \mathds{1}\,.
\esp\eeq
Inserting the completeness relations into the expressions for the (co-)homology intersection pairing  yields the \textit{twisted Riemann bilinear relations} (TRBRs) that relate the four matrices $\bP$, $\bs{\check{P}}$, $\bC$ and $\bH$\cite{Cho_Matsumoto_1995}
\beq\bsp
\label{generalriemann}
\frac{1}{(2\pi i)^n} \bs{P} \left(\bs{H}^{-1}\right)^T \bs{\check{P}}^T &\, =    \bs{C}\,,  \\ 
\frac{1}{(2\pi i)^n} \bs{P}^T \left(\bs{C}^{-1}\right)^T \bs{\check{P}}&\, =     \bs{H} \, . 
\esp\eeq
Equations  (\ref{eq:completeness}) and (\ref{generalriemann})  also hold for the relative case discussed in the next section. 

We have already mentioned that eqs.~\eqref{eq:Gauss-Manin} and~\eqref{eq:dC} together imply that the differential $\rdex$ with respect to the external parameters $\bx$ should not receive contributions from the twisted cycles. Indeed, the TRBRs together with eqs.~\eqref{eq:Gauss-Manin} and~\eqref{eq:dC} imply
\beq\bsp
\rdex\bs{H} &\,= \frac{1}{(2\pi i)^n} \left[\big(\rdex\bs{P}^T\big) \left(\bs{C}^{-1}\right)^T \bs{\check{P}} + \bs{P}^T \rdex\!\left(\bs{C}^{-1}\right)^T \bs{\check{P}} +  \bs{P}^T \left(\bs{C}^{-1}\right)^T \rdex\big(\bs{\check{P}}\big)\right]\\
&\,=0\,.
\esp\eeq

\begin{example}[Gauss' hypergeometric ${{}_2F_1}$ function]
\label{example11}
As a simple, but illustrative example, we consider Gauss' hypergeometric function, defined by 
\begin{align}
\label{eq_f21}
    {}_2F_1(a,b;c;y)=\frac{\Gamma(c)}{\Gamma(a)\Gamma(c-a)} \int_0^1 z^{a-1} (1-z)^{c-a-1} (1-yz)^{-b}\rd z\, .
\end{align} 
These integrals can be interpreted as periods of a twisted cohomology theory with the twist defined by the product in the integrand. For the example discussed here, we choose 
\begin{align}
\label{eq_twistF21}
    \Phi = z^{ \alpha_0 } (1-z)^{\alpha_1} \left(1-y{z}\right)^{\alpha_{x}} \,,
\end{align}
with $0<\alpha_i<1$. In this case, $X=\mathbb{C}-\{0,1,x\}$, with $x:=y^{-1}$. We choose as a basis for $H_{\text{dR}}^1(X,\nabla)$ the (classes of the) 1-forms:
\beq\bsp
\label{herenotchange}
\varphi_1&\,=\frac{(1-x)\rd z}{(1-z)\left(z-x\right)}\,,\\
\varphi_2&\,=- \frac{\rd z}{z} \,.
\esp\eeq
For the dual basis, we choose
\beq\bsp
\label{eq:2F1_non-rel_dual}
\check{\varphi}_1&\,=\left[\frac{(1-x)\rd z}{(1-z)\left(z-x\right)}\right]_c \, ,\\
\check{\varphi}_2&\,=- \left[ \frac{\rd z}{z}\right]_c \,. 
\esp\eeq
For the basis of contours in $H_1(X,\check{\mathcal{L}})$ we choose the regularised versions of the dual basis elements $\check{\gamma}_1=[x,\infty] \text{ and } \check{\gamma}_2=[0,1]$.
Then, for $x>1$, the period matrix $\bs{P}$ has the entries 
\begin{align}
\nonumber
 P_{11}(x,\balpha) &= \left(1-x^{-1}\right)\,e^{i\pi \alpha_{1x}}\, x^{\alpha_{01} }\, \tfrac{\Gamma(\alpha_{x} )\Gamma(1-\alpha_{01x} )}{\Gamma(1-\alpha_{01} 
)}\,{}_2F_1\left(1-\alpha_1 , -\alpha_{01x} +1;1-\alpha_{01} ;y\right) ,\\
\label{firstperiod}
P_{12}(x,\balpha)&= \left(1-x^{-1}\right)\tfrac{\Gamma(1+\alpha_0 )\Gamma(\alpha_1 )}{\Gamma(1+\alpha_{01} )} {}_2F_1(\alpha_0  +1, 1-\alpha_{x}  ; 1+\alpha_{01} ;y)\,,\\\nonumber
P_{21}(x,\balpha) &=- e^{i\pi \alpha_{1x} } x^{\alpha_{01} } \tfrac{\Gamma(1+\alpha_{x} ) \Gamma(-\alpha_{01x}   )}{\Gamma(1-\alpha_{01} )} \,{}_2F_1(-\alpha_1 , -\alpha_{01x} ; 1-\alpha_{01}  ; y)\,,\\
\nonumber
P_{22}(x,\balpha) &= - \tfrac{\Gamma(\alpha_0  ) \Gamma(1+\alpha_1 )}{\Gamma(1+\alpha_{01} )}\, {}_2F_1(\alpha_0 , -\alpha_{x}  ;1+\alpha_{01} ; y)\,,
\end{align}
where we use the notation $\alpha_{ij\dots} = \alpha_i+\alpha_j+\dots$. Our choice of dual basis allows us to write the dual period matrix as
\begin{align}\label{eq:non_rel_period_mat}
    \check{\bs{P}}(x,\balpha)=\bs{P}(x,-\balpha)\, . 
\end{align} 
The intersection matrices are 
\beq\bsp
\label{cmat}
\bs{C}(x,\balpha)&\,=\left(\begin{smallmatrix} -\tfrac{ 1}{\alpha_0 }+\frac{1}{\alpha_{01x}} & 0 \\ 0&-\tfrac{\alpha_{1x}}{\alpha_1\alpha_x  }
    \end{smallmatrix}\right)\,, \\
    \bs{H}(\balpha)&\,=\tfrac{1}{2\pi i}\left(\begin{smallmatrix}
        \pi \csc(\pi \alpha_x )\sin(\pi \alpha_{01}  ) \csc(\pi \alpha_{01x} )  &0 \\ 0 &\pi\cot(\pi \alpha_0 )+\pi\cot(\pi \alpha_1 )
    \end{smallmatrix}\right)\, . 
\esp\eeq
The TRBRs provide quadratic relations among the entries of the period matrix, such as 
\beq\bsp
\frac{P_{11}(x,\bs{\alpha})\,P_{11}(x,-\bs{\alpha})}{\cot(\pi\alpha_0)+\cot(\pi\alpha_1)}+\frac{\sin(\pi\alpha_x)\sin(\pi\alpha_{01x})}{\sin(\pi\alpha_{01})}P_{12}(x,\bs{\alpha})\,P_{12}(x,-\bs{\alpha}) &\,= \pi \,\left(\frac{1}{\alpha_{01x}}-\frac{1}{\alpha_0}\right) \,,\\
\frac{P_{11}(x,\bs{\alpha})\,P_{21}(x,-\bs{\alpha})}{\cot(\pi\alpha_0)+\cot(\pi\alpha_1)} + \frac{\sin(\pi\alpha_x)\sin(\pi\alpha_{01x})}{\sin(\pi\alpha_{01})} P_{12}(x,\bs{\alpha})\,P_{22}(x,-\bs{\alpha}) &\,=0\,.
\esp\eeq
\end{example}

\subsection{Relative twisted (co-)homology}
\label{subsec:reltwist}
If condition~\eqref{restrict} does not hold, we need to work with \textit{relative} twisted cohomology groups~\cite{Caron-Huot:2021iev,Caron-Huot:2021xqj}, where we can  treat integrands $\Phi \varphi$ with poles not `regulated' by non-integer exponents in $\Phi$. 
The relative dual cycles are allowed to have boundaries in a subcomplex  consisting of the unregulated boundaries, and the relative \textit{dual} co-cycles have additional terms  that keep track of boundary contributions. We give a very brief review here. For more details on this framework see refs.~\cite{matsumoto_relative_2019-1,Caron-Huot:2021xqj, Caron-Huot:2021iev}, and for applications in physics, see refs.~\cite{Caron-Huot:2021xqj, Caron-Huot:2021iev,Giroux:2022wav,De:2023xue,Bhardwaj:2023vvm,Brunello:2023rpq,Crisanti:2024onv}.

\paragraph{Relative twisted homology.} 
For a fixed integral \eqref{gen_twisted_integral}, in which we are interested, we define the following sets of hypersurfaces\footnote{By our convention, all possible poles of the differential forms $\varphi$ we consider must be either zeroes or poles of the twist, even if that means including $L_j(\bz)^0$ as a factor into the twist.}:
\beq\bsp
    \Sigma &=\{\bs{z}\, |\, L_i(\bs{z}) = 0 \wedge \alpha_i \notin\mathbb{Z}\} \cup \{\infty\}\,, \\
    D_+&=\{ \bs{z}\, |\, \Phi\varphi(\bs{z}) = 0 \}\,, \\
    D_- &= \{\bs{z}\, |\, \bs{z} \text{ is a pole of }\Phi\varphi(\bs{z})\textrm{~and~}  \bs{z} \notin \Sigma\} \, . 
\esp\eeq
Additionally, we denote:
\beq\bsp
    X_{\pm}&= \mathbb{C}^n- \Sigma -D_{\pm}\, . 
\esp\eeq
Since $\Phi\varphi$ vanishes on $D_+$, we can work relative to this set, i.e. consider cycles with boundaries on $D_+$ as those cycles are closed.
That means, we define \textit{relative twisted cycles} as elements of
\begin{align}
\label{relativecycels}
    C_n(X_-,D_+,\check{\mathcal{L}}) = C_n(X_-,\check{\mathcal{L}})/C_n(D_+,\check{\mathcal{L}}) \, ,   
\end{align}
with
\beq
C_n(D_+,\check{\mathcal{L}})=C_n(X_-,\check{\mathcal{L}}) |_{D_+}\,. 
\eeq
Analogous to (\ref{CBeq}), we define 
\begin{align}
B_n(X_-,D_+,\check{\mathcal{L}})=B_n(X_-,D_+,\check{\mathcal{L}})/B_n(D_+,\check{\mathcal{L}})\, . 
\end{align}
The \textit{relative twisted homology group} is
\begin{align}
    H_n(X_-,D_+,\check{\mathcal{L}}) =  C_n(X_-,D_+,\check{\mathcal{L}})/  B_n(X_-,D_+,\check{\mathcal{L}})\, .
\end{align}
As a basis for the homology group $H_1(X_-,D_+,\check{\mathcal{L}})$ we generally choose a combination of loops around the poles in $D_{-}$ and (regulated) chambers between the branchpoints $\{L_i(\bs{z})=0\}$. The dual relative twisted homology group is defined by
\begin{align}
H_n(X_+, D_-,\mathcal{L}) = C_n(X_+,D_-,\mathcal{L})/ B_n(X_+,D_-,\mathcal{L}) \, . 
\end{align}
Note, that the dual twisted homology group is relative to the poles in $D_-$. Thus, the (basis) elements of the dual twisted homology group can end at poles of $\Phi\varphi$. 

\paragraph{Relative twisted cohomology.}
Similarly, one defines the relative twisted cohomology group
\begin{align}
\label{relativetwistedcohojmology}
    H^n_{\text{dR}}(X_-,D_+,\nabla) = C^n(X_-,D_+, \nabla)/ B^n(X_-,D_+,\nabla) \, ,
\end{align}
where the relative twisted co-cycles are elements of $ C^n(X_-,D_+, \nabla)= C^n(X_-,\nabla)/C^n(D_+, \nabla)$. Note that for cases where $D_+=\varnothing$ we have \cite{matsumoto_relative_2019-1}
\begin{align}\label{eq:H_dR_rel-to-non_rel}
    H_{\text{dR}}^n (X_-,D_+,\nabla)\cong H_{\text{dR}}^n (X_-,\nabla)\, . 
\end{align}
This applies to all examples considered in this paper. Note however that $X_{-}\neq X$ with $X$ as defined in the non-relative case, but rather $X_{-}=X-D_{-}$. One needs to be more careful when considering the dual relative twisted cohomology group 
\begin{align}
    H^n_{\text{dR}}(X_+,D_-,\check{\nabla}) = C^n(X_+,D_-,\check{\nabla})/B^n(X_+,D_-,\check{\nabla}) \, . 
\end{align}
and more importantly, its version with compact support, which is needed for a well-defined intersection pairing. Elements of $H^n_{\text{dR}}(X_+,D_-,\check{\nabla})$ have the form  
\begin{align}
    \check{\varphi} = \theta \psi + \delta_1(\theta \psi_1)+\dots + \delta_{1,2} \left(\theta \psi_{1,1}\right)+ \dots \, .
\end{align}
Here $\theta$ is a symbol that keeps track of possible boundary terms and the sum is over all (intersections of) sub-boundaries of $D_-$. The form $\delta_{i_1,\dots,i_p}({\phi}_{i_1,\dots,i_p})$ is the \textit{Leray coboundary} of a form ${\phi}_{i_1,\dots,i_p}$  living on the boundary $ \{D_{i_1}=0\}\cup \dots \cup\{ D_{i_p}=0\}$. It can be explicitly represented as 
\begin{align}\label{reldelta}
\delta_{i_1,\dots,i_p}({\phi}_{i_1,\dots,i_p})=\frac{\Phi}{\Phi|_{{i_1,\dots,i_p}}} \rd\theta_{i_1}\wedge\dots \rd\theta_{i_p} \wedge {\phi}_{i_1,\dots,i_p}\,,
\end{align}
where we may interpret $\rd\theta$ as the derivative of the Heaviside step function,
\beq\label{eq:heaviside}
\theta(x) = \left\{\begin{array}{ll}
1\,,&  x >0\,,\\ 0\,,& x \le 0\,.
\end{array}\right.
\eeq
Details on the construction and meaning of $\rd \theta$ can be found in refs.~\cite{Caron-Huot:2021iev,Caron-Huot:2021xqj,Giroux:2022wav}. For one-forms, this construction implies that a dual basis of $H^1_{\text{dR}}(X_+,D_-,\check{\nabla})$ is given by
\begin{align}
\label{zibasis}
\delta_{z_i}(1)=\frac{\Phi}{\Phi|_{z_i=0}} \rd\theta(z-z_i)\,,  \text{~~~for~} z_i \text{ a pole}\,,
\end{align}
and compactifications of twisted co-cycles ${\phi}_{\text{reg}}$  with `regulated' singularities, 
\begin{align}
\label{regBasis}
\left[\phi_\text{reg}\right]^{\text{rel}}_c=\phi_\text{reg} \prod_i \theta(z-z_i) +\sum_i\psi_i \,\rd \theta(z-z_i)\,,
\end{align}
with ${\phi}_{\text{reg}}$ chosen as in the non-relative case and $z_i \in D_+\cup \Sigma$.
The functions $\psi_i$ are the local primitives defined by the twist as
$\check{\nabla}\psi_i =\phi_{\text{reg}}$. With the basis choice as in eqs.~\eqref{zibasis},\eqref{regBasis}, the intersection numbers can be computed via residues:
\beq\bsp
\langle \varphi|\delta_{z_i}(1)\rangle&=\text{Res}_{z=z_i}\left(\frac{\Phi}{\Phi|_{z=z_i}}\varphi\right)\,,\\
\langle\varphi|\left[\phi_\text{reg}\right]_c\rangle&=\sum_k \text{Res}_{z=z_k}(\psi_k \varphi)\, .
\esp\eeq
\begin{example}[Gauss hypergeometric ${}_2F_1$ function in relative twisted cohomology]
\label{example21}

We take again the ${}_2F_1$ function with the twist in eq.~\eqref{eq_twistF21} as an example, but this time we consider the case with a pole at $z=x$ by taking $\alpha_x\rightarrow 0$, such that the twist is now given by
\begin{align}
\label{eq_reltwist}
    \Phi= z^{\alpha_0 }(1-z)^{\alpha_1 }\left(1-x^{-1}z\right)^{0}\, . 
\end{align}
Since the twist does not contain factors raised to positive integer powers, eq.~\eqref{eq:H_dR_rel-to-non_rel} holds, and we can pick as a basis for  $ H_{\text{dR}}^n (X_-,D_+,\nabla)\cong H_{\text{dR}}^n (X_-,\nabla)$ the same basis as in the non-relative case in eq.~\eqref{herenotchange}.
We pick as a basis of dual forms
     \beq\bsp\label{eq_relabasis}
\check{\varphi}^{\text{rel}}_1&\,=-\delta_x(1)= -\Phi \left(\Phi|_{z= x}\right)^{-1} \rd \theta(z-x)\,,\\
\check{\varphi}^{\text{rel}}_2&\,=-\left[\frac{\rd z}{z}\right]^{\text{rel}}_c\, . 
\esp\eeq
Note that $\check{\varphi}^{\text{rel}}_2$ agrees with the dual basis element $\check{\varphi}_2$ in the non-relative case in eq.~\eqref{eq:2F1_non-rel_dual}, because the pole of $\varphi_2$ at $z=0$ is regulated by the twist. Since the pole in $\varphi_1$ is not regulated by the twist, the dual basis element $\check{\varphi}^{\text{rel}}_1$ differs from its non-relative counterpart in eq.~\eqref{eq:2F1_non-rel_dual}.

As a basis of twisted cycles, we choose loaded cycles supported on:
\beq\bsp
    \label{cycleexample1}
    \gamma_1&\,= (2\pi i)^{-1}S_\eta(x)\,,\\
   \gamma_2&\,=[0,1]_{\text{reg}}\,  , 
\esp\eeq
where $S_\eta(x)$  is the circle  of radius $\eta$ and center $x$. The first cycle differs from the choice in the non-relative case in example~\ref{example11}, where all basis cycles are regularised intervals, because the function $\Phi$ is not multi-valued at $z=x$, and thus the circle defines a closed contour. For the dual cycles, we choose as a basis the loaded cycles supported on 
\beq\bsp
    \label{cycleexample2}
    \check{\gamma}_1&\,=[x,\infty]\,,\\
     \check{\gamma}_2&\,=[0,1]\,  . 
\esp\eeq
  The cycle $\check{\gamma}_1$ can end at the pole $x$, since we are working relative to it. 
With this choice of basis, the period matrix for $x>1$ is given by 
\beq\bsp
    \label{pmatrel}
P_{11}^{\text{rel}}(x,\balpha)&\,= x^{\alpha_0 } (x-1)^{\alpha_1 } e^{\pi i \alpha_1}\,,\\
P_{12}^{\text{rel}}(x,\balpha)&\,=-\tfrac{1-x}{x}\, \tfrac{\Gamma(1+\alpha_0)\Gamma(\alpha_1 )}{\Gamma(1+\alpha_{01} )}\,{}_2F_1\left(1+\alpha_0 ,1;1+\alpha_{01} ;y\right)\,,\\
P_{21}^{\text{rel}}(x,\balpha)&\,= 0\,,\\
P_{22}^{\text{rel}}(x,\balpha)&\,= -\tfrac{\Gamma(\alpha_0 )\Gamma(1+\alpha_1 )}{\Gamma(1+\alpha_{01} )} \, ,
\esp\eeq
and the dual period matrix is 
\beq\bsp
\label{pmatreldual}
\check{P}^{\text{rel}}_{11}(x,\balpha) &\,=- x^{-\alpha_0 } (x-1)^{-\alpha_1 } e^{-\pi i \alpha_1}\,,\\
\check{P}^{\text{rel}}_{12}(x,\balpha) &\,= 0\,, \\
\check{P}^{\text{rel}}_{21}(x,\balpha)&\,=-e^{-i\pi  \alpha_1 }   \,   \tfrac{x^{-\alpha_{01} }}{\alpha_{01}}\, {}_2F_1\left(\alpha_{01} , \alpha_1 ;1+ \alpha_{01} ;y\right)\,,\\
\check{P}^{\text{rel}}_{22}(x,\balpha)&\,=-\tfrac{\Gamma(-\alpha_0 )\Gamma(1-\alpha_1 )}{\Gamma(1-\alpha_{01} )} \, . 
\esp\eeq
The intersection matrices are
\beq\bsp
\label{intersectionsrelative}
     \bs{C}^{\text{rel}}(x,\balpha)&\,=\left(\begin{smallmatrix}
     -1&0\\ 0 & -\frac{\alpha_1}{\alpha_0 \alpha_{01}}
    \end{smallmatrix} \right)\,,\\
  \bs{H}^{\text{rel}}(\balpha)&\,=\tfrac{1}{2\pi i}\,\left(\begin{smallmatrix}
    1&0\\ 0 & \pi\cot\left(\pi  \alpha_0\right)+\pi\cot\left(\pi  \alpha_1\right)
    \end{smallmatrix}\right) \,. 
\esp\eeq

Note that $\bs{\check{P}}^{\text{rel}}(x,\balpha) \neq \bP^{\text{rel}}(x,-\balpha)$. Since this observation will be crucial when we discuss Feynman integrals, let us spend some time to discuss this point. More precisely, one may wonder if we could change basis for the twisted cycles and co-cycles so that $\bs{\check{P}}^{\text{rel}}(x,\balpha) = \bP^{\text{rel}}(x,-\balpha)$. In the following, we show that this is not possible.
Clearly, the only obstructions come from $\check{P}^{\text{rel}}_{21}(x,\balpha)$ and $\check{P}^{\text{rel}}_{12}(x,\balpha)$. We now argue that it is not possible to change basis such that $\check{P}^{\text{rel}}_{21}(x,\balpha),\check{P}^{\text{rel}}_{12}(x,\balpha)$ can be expressed in terms of the entries of $\bs{{P}}^{\text{rel}}(x,-\balpha)$. Indeed, if that was the case, we could find rational functions $c_{ij}(x,\balpha)$ such that for example
\beq\label{independence}
\check{P}^{\text{rel}}_{12}(x,\balpha) = \sum_{i,j=1}^2 c_{ij}(x,\balpha)\,{P}^{\text{rel}}_{ij}(x,-\balpha)\,.
\eeq
It is sufficient to show that this is impossible if we expand all the functions around $\balpha =0$ (for example, we may let $\alpha_i=a_i\eps$, and then expand in $\eps$). The expansion of the hypergeometric functions can be obtained from {\sc HypExp}~\cite{Huber:2005yg,Huber:2007dx}. One may then observe that the coefficients in the Laurent expansion of $\bs{{P}}^{\text{rel}}(x,\balpha)$ and $\bs{\check{P}}^{\text{rel}}(x,\balpha)$ are pure~\cite{Arkani-Hamed:2010pyv} combinations of harmonic polylogarithms~\cite{Remiddi:1999ew} of uniform weight. Since the transcendental functions that appear in the expansion are pure, it follows that the $c_{ij}(x,\balpha)$ must be independent of $x$, and the coefficients in the expansion must also be pure constants of uniform weight. Moreover, it is possible to write down a basis for harmonic polylogarithms~\cite{DDMS}. Then by comparing the first few orders in the expansion written in this basis, one can show that eq.~\eqref{independence} can never be satisfied. Hence, we have shown that there is no basis in which $\bs{\check{P}}^{\text{rel}}(x,\balpha) \neq \bP^{\text{rel}}(x,-\balpha)$, as claimed.



\end{example}

\section{Quadratic relations for Feynman integrals and their cuts}
\label{sec:RelationsRiemann}

In this section, we explain how to interpret dimensionally-regulated Feynman integrals and their cuts as twisted periods, following ref.~\cite{Mizera:2017rqa}. We then use this framework to investigate TRBRs for (cut) Feynman integrals, which relate the (dual) period matrix to the intersection matrices $\bC$ and $\bH$. Methods to compute intersection numbers between twisted co-cycles are discussed, for example in refs.~\cite{Mastrolia:2018uzb,Frellesvig:2019kgj,Frellesvig:2020qot,Weinzierl:2020xyy,Caron-Huot:2021xqj,Caron-Huot:2021iev,Chestnov:2022xsy,Cacciatori:2022mbi,Brunello:2023rpq,Chestnov:2022alh,Crisanti:2024onv}. The question of how to fix a basis of cycles, or the equivalent question of how to determine a period matrix given a basis of differential forms for master integrals, is much less explored. We therefore start by defining a period matrix for a given family of Feynman integrals in section~\ref{sec:period_matrix}, before we discuss TRBRs for Feynman integrals in section~\ref{sec:TRBRs-FI}.




\subsection{A period matrix for Feynman integrals}
\label{sec:period_matrix}


Consider a family of Feynman integrals in a fixed dimension $D$. Then each member of the family is defined by the vector $\bs{\nu}$ of exponents, and we can write (cf.~eq.~\eqref{eq_Baikov2} and~\eqref{Baikov_LL})\footnote{For readibility, we relabelled $N'$ and $\mathcal{C}'$ into $N$ and $\mathcal{C}$ compared to eq.~\eqref{Baikov_LL}.}
\beq
\hat{I}^D_{\bs{\nu}} = \int_{\mathcal{C}}\Phi\varphi_{\bs{\nu}}\,,
\eeq
with the twist given by the (product of) Baikov polynomial(s):
\beq\label{eq:Feynman_twist}
\Phi= \mathcal{B}_1(\bs{z})^{\mu_1}\dots \mathcal{B}_K(\bs{z})^{\mu_K}\,.
\eeq
The exponents have the form
\beq\label{eq:Baikov_twis_exponent}
\mu_i = \tfrac{m_{i}}{2}+\sigma_i\eps\,,\qquad m_{i}\in\mathbb{Z}, \, \, \sigma_i\in\{\pm 1\} \, ,
\eeq 
and the rational differential form is 
\beq
\varphi_{\bs{\nu}} = \rd^{N}\!z\,\prod_{i=1}^{N}z_i^{-\nu_i}\,.
\eeq
From this it follows that Feynman integrals and their cuts are twisted periods.\footnote{We note that it is also possible to work directly in momentum space in order to identify Feynman integrals as twisted periods, see, e.g., ref.~\cite{Caron-Huot:2021xqj}.}

Let us now discuss how we can define a period matrix for this family of Feynman integrals. Parts of the results were already presented and/or conjectured in some form in refs.~\cite{Abreu:2019eyg,Bonisch:2021yfw}, but this is the first time a derivation is presented.  
 The following concept is well known from the study of IBP relations: each twisted co-cycle $\varphi_{\bs{\nu}}$ belongs to a \emph{sector}, characterised by an element $\bTheta_{\bnu}\in \{0,1\}^N$ defined by
\beq
\bTheta_{\bnu} = \big(\theta(\nu_1),\ldots,\theta(\nu_N)\big)\,,
\eeq
where $\theta(\nu)$ is the Heaviside step function from eq.~\eqref{eq:heaviside}.
There is a natural partial order on sectors given by
\beq
\bTheta_{\bnu_1}\succcurlyeq\bTheta_{\bnu_2} \textrm{~~~if~~~} \theta(\nu_{1i}) \ge \theta(\nu_{2i}) \textrm{~~for all~~} i\,.
\eeq
This partial order induces a partial order on the master integrals. A sector is called \emph{reducible} if every twisted co-cycle belonging to this sector can be expressed as a linear combination of co-cycles from a lower sector. A \emph{top-sector} is an irreducible sector that is maximal for the partial ordering. Note that we can always choose a basis of master integrals so that there is no master integral that belongs to a reducible sector. In the following, we assume that we have chosen such a basis. We denote the irreducible sectors of the family by $\bTheta_1, \ldots, \bTheta_S$, and $M_i$ is the number of master integrals in that irreducible sector. Note that we always have $|\bTheta_i|\ge L$, where $|\bTheta_i|$ denotes the number of non-zero entries of $\bTheta_i$. This is a direct consequence of the fact that in dimensional regularisation all scaleless integrals vanish. We emphasise here an important point: The list of irreducible sectors determines the set of propagators that enter the master integrals, but we cannot determine a priori how to distinguish master integrals within a given sector. At this point, this can only be done via a case by case analysis, e.g., by solving the IBP relations.

Let us now explain how we can construct a twisted period matrix for this family in terms of cut integrals. We start from a top-sector, say $\bTheta_1$. It is generated by $M_1$ master integrals which share exactly the same propagators, but possibly raised to different positive powers and/or different numerator factors. We can identify $M_1$ different maximal cuts, or equivalently $M_1$ independent contours which encircle the propagators of this sector. There is no general algorithm of how to choose these $M_1$ integration cycles, other than that they need to encircle the poles defined by the propagators of that sector. This is similar to what we discussed for the master integrals: we cannot a priori characterize the master integrals in a given sector, other than that they need to have a precise set of propagators. In practice, they can often be identified by analyzing the zeroes of the Baikov polynomials in a given sector, cf.,~e.g.,~refs.~\cite{Primo:2016ebd,Primo:2017ipr,Frellesvig:2017aai,Bosma:2017ens}. We can evaluate all master integrals (including those from lower sectors) on the cycles obtained in this way, and we obtain zero for all master integrals that are not in the sector $\bTheta_1$. We obtain in this way $M_1$ vectors. These vectors are linearly independent because the maximal cuts are independent. Moreover, they have at most $M_1$ non-zero entries. We repeat this procedure for all other top-sectors, say $\bTheta_2,\ldots, \bTheta_k$ (with $k$ the number of top-sectors). Note that vectors for two different sectors are necessarily linearly independent because they have non-zero entries in different places.

Next, we turn to an irreducible next-to-top-sector, say $\bTheta_{k+1}$. There are $M_{k+1}$ independent maximal cut contours for this sector. We can again evaluate all master integrals on these $M_{k+1}$ cycles. Unlike for the top-sectors, however, we obtain zero for all integrals from lower sectors $\bTheta_l\preccurlyeq\bTheta_{k+1}$, but the result is not necessarily zero for the top-sectors. The main point is that we can use exactly the same contours that define independent maximal cuts in the sector $\bTheta_{k+1}$ to define non-maximal cuts in higher sectors. This is particularly manifest in the loop momentum representation in eq.~\eqref{eq_sec2.1}, because the integration measure only depends on the number of loops. The same conclusion can easily be reached from the Baikov representation. We obtain in this way $M_{k+1}$ linearly independent vectors (because we have chosen a set of independent maximal cuts), which have zero entries for all master integrals from sectors $\bTheta_l\preccurlyeq\bTheta_{k+1}$. Note that these vectors are linearly independent from the vectors constructed from the top-sectors, because the latter have zeroes in all entries corresponding to the sector $\bTheta_{k+1}$. 

We can continue in this way until we have exhausted all sectors. For each irreducible sector $\bTheta_r$ we obtain in this way a set of $M_r$ independent maximal cut contours, from which we can construct $M_r$ linearly independent vectors. Moreover the vectors obtained from two different sectors are linearly independent. We can put these column vectors together to form an $M\times M$ matrix $\bP$, with $M=\sum_iM_i$ the total number of master integrals. This matrix has the following properties:
\begin{enumerate}
\item $\bP$ has full rank, because all columns are linearly independent.
\item By reverse-unitarity, each column satisfies the differential equation~\eqref{eq:DEQ_generic} for this basis of master integrals.
\item If the master integrals and the maximal cut contours are ordered in a way that respects the natural partial ordering on the sectors, then $\bP$ is block upper-triangular. The blocks on the diagonal are the maximal cuts for the irreducible sectors, while the entries above the diagonal are non-maximal cuts.
\item All non-zero entries are cut Feynman integrals.
\end{enumerate}
The first two properties identify the matrix $\bP$ as a fundamental solution matrix to eq.~\eqref{eq:DEQ_generic}, and so $\bP$ is the period matrix for this particular choice of bases of twisted cycles and co-cycles.  Note that, since $\bP$ is the fundamental solution matrix of eq.~\eqref{eq:DEQ_generic}, every solution to eq.~\eqref{eq:DEQ_generic} is of the form $\bP\bIhat_0$, where $\bIhat_0$ is a constant vector of initial conditions (it still depends on the dimensional regulator $\eps$). This implies that all the full, uncut, Feynman integrals of this family can be written as a linear combination of the cuts that enter the period matrix. In particular, there is a constant vector $\bIhat_0^{\textrm{uncut}}$ such that the vector of master integrals $\bIhat^{\textrm{uncut}}$ (evaluated on the contour that define the uncut integral) is given by $\bIhat^{\textrm{uncut}}=\bP\bIhat_0^{\textrm{uncut}}$. As a consequence, every Feynman integral is a linear combination of its cuts, similarly to the statement of the celebrated Feynman tree theorem~\cite{osti_4077804,Feynman1972}.

As a corollary of our construction of a period matrix, we have obtained a way to identify a basis for the twisted homology group associated to a Feynman integral: A basis of twisted cycles is obtained by considering a set of independent maximal cut contours for each irreducible sector. This basis has two interesting properties: Since $|\bTheta_i|\ge L$, our basis consists of contours that correspond to cutting at least $L$ propagators. Moreover, it is easy to check that every loop of the underlying Feynman graph always contains at least one cut propagator. The existence of a spanning set of cuts with this property was conjectured in ref.~\cite{Abreu:2019eyg}.

\begin{example}[The period matrix for the one-loop bubble integral]
\label{ex:bubble_period_matrix}
We consider as an example the massive bubble integral in $D=2-2\eps$ dimensions:
\begin{align}
I_{\nu_1\nu_2}^{D} =e^{\gamma_E\varepsilon} \int\frac{\rd^D\ell}{i\pi^{\frac{D}{2}}}\frac{1}{(\ell^2-m_1^2)^{\nu_1}((\ell-p)^2-m_2^2)^{\nu_2}}\, . 
\end{align}
There are three irreducible sectors:
\beq
\bTheta_1=(1,1)\,,\qquad\bTheta_2=(1,0)\,,\qquad\bTheta_3=(0,1)\,.
\eeq
Each sector has one master integral, namely $I_{1,1}^{D}$, $I_{1,0}^{D}$ and $I_{0,1}^{D}$. The bubble integral $I_{1,1}^{D}$ can be expressed in terms of Appell $F_4$ functions, which can be reduced to Gauss' hypergeometric function evaluated at an algebraic argument, cf., e.g.,~ref.~\cite{Anastasiou:1999ui}. The one-loop tadpole integral is given by
\beq
I_{1,0}^{D} =-\frac{e^{\gamma_E \epsilon}\Gamma(1+\epsilon)(m_1^2)^{-\epsilon}}{\epsilon}\,,
\eeq
and the result for $I_{0,1}^{D}$ is obtained by exchanging $m_1$ and $m_2$.

We can write down a period matrix for the one-loop bubble integral as follows:
\beq\label{eq:bubble_per_mat}
{\bs{P}}_{\tikz[baseline=-0.5ex]{\draw (0,0) circle [radius=2pt];
		\draw (-4pt,0) -- (-2pt,0);
		\draw (2pt,0) -- (4pt,0);}}
=
\begin{pmatrix}
	\tikz[baseline=-0ex]{
		\draw (0,0) circle [radius=5pt];
		\draw (-10pt,0) -- (-5pt,0);
		\draw (5pt,0) -- (10pt,0);
		\draw [red, thick] (0,2pt) -- (0,8pt); 
		\draw [red, thick] (0,-2pt) -- (0,-8pt); 
		\node[font=\footnotesize] at (0,11pt) {$m_1$}; 
		\node[font=\footnotesize] at (0,-12pt) {$m_2$}; 
	}
	& \hspace{0.2cm}
	\tikz[baseline=-0ex]{
		\draw (0,0) circle [radius=5pt];
		\draw (-10pt,0) -- (-5pt,0);
		\draw (5pt,0) -- (10pt,0);
		\draw [red, thick] (0,2pt) -- (0,8pt); 
		\node[font=\footnotesize] at (0,11pt) {$m_1$}; 
		\node[font=\footnotesize] at (0,-12pt) {$m_2$}; 
	}
	&\hspace{0.2cm}
	\tikz[baseline=-0ex]{
		\draw (0,0) circle [radius=5pt];
		\draw (-10pt,0) -- (-5pt,0);
		\draw (5pt,0) -- (10pt,0);
		\draw [red, thick] (0,-2pt) -- (0,-8pt); 
		\node[font=\footnotesize] at (0,11pt) {$m_1$}; 
		\node[font=\footnotesize] at (0,-12pt) {$m_2$}; 
	}
	\\
	0&   \hspace{0.2cm}                                                
	\tikz[baseline=-1ex]{
		\draw (0,0) circle [radius=5pt];
		\draw (0,-5pt) -- (0,-10pt); 
		\draw [red, thick] (0,2pt) -- (0,8pt); 
		\node[font=\footnotesize] at (0,11pt) {$m_1$}; 
	}
	&\hspace{0.2cm}0\\
	\displaystyle 0&\hspace{0.2cm}    0&\hspace{0.2cm}                                          
	\tikz[baseline=-1ex]{
		\draw (0,0) circle [radius=5pt];
		\draw (0,-5pt) -- (0,-10pt); 
		\draw [red, thick] (0,2pt) -- (0,8pt); 
		\node[font=\footnotesize] at (0,11pt) {$m_2$}; 
	}
\end{pmatrix}.
\eeq
    Analytic results for all the cut integrals entering the period matrix can be found in ref.~\cite{Abreu:2019eyg}.

\end{example}

\begin{example}[The period matrix for the unequal-mass sunrise integral]
As a second example we consider the unequal-mass sunrise integral in $D=2-2\varepsilon$ dimensions
\begin{align}
\label{sunexstart}
I_{\bs{\nu}}^{{\tikz[baseline=-0.5ex]{\draw (0,0) circle [radius=2pt];
                             \draw (-4pt,0) -- (4pt,0);}}}
\left(p^2,m_1^2,m_2^2,m_3^2 \right)= e^{2\gamma_E \varepsilon} \int \frac{\rd^D \ell_1}{i\pi^{D/2}} \frac{\rd^D \ell_2}{i\pi^{D/2}}\frac{1}{D_1^{\nu_1}D_2^{\nu_2}D_3^{\nu_3}}\, ,
\end{align}
with 
\beq\bsp
    D_1&\,=\ell_1^2-m_1^2\,,\\
     D_2&\,=(\ell_2-\ell_1)^2-m_2^2\,,\\
D_3&\,=(p+\ell_2)^2-m_3\, . 
\esp\eeq
This integral family has four irreducible sectors
\begin{equation}
    \bTheta_1=(1,1,1),\quad \bTheta_2=(0,1,1),\quad \bTheta_3=(1,0,1), \quad\bTheta_4=(1,1,0)\, ,
\end{equation}
with lower sectors vanishing since they correspond to scaleless integrals. The subsectors $\bTheta_2,\bTheta_3,\bTheta_4$ each have one master integral and are simply given by a product of two one-loop tadpole integrals. The top sector $\bTheta_1$ is a more complicated and admits four independent maximal cut contours and hence also four different master integrals. Following the prescription laid out above we thus find the following period matrix:
\begin{figure}[H]
    \centering
    \includegraphics[width=0.5\linewidth]{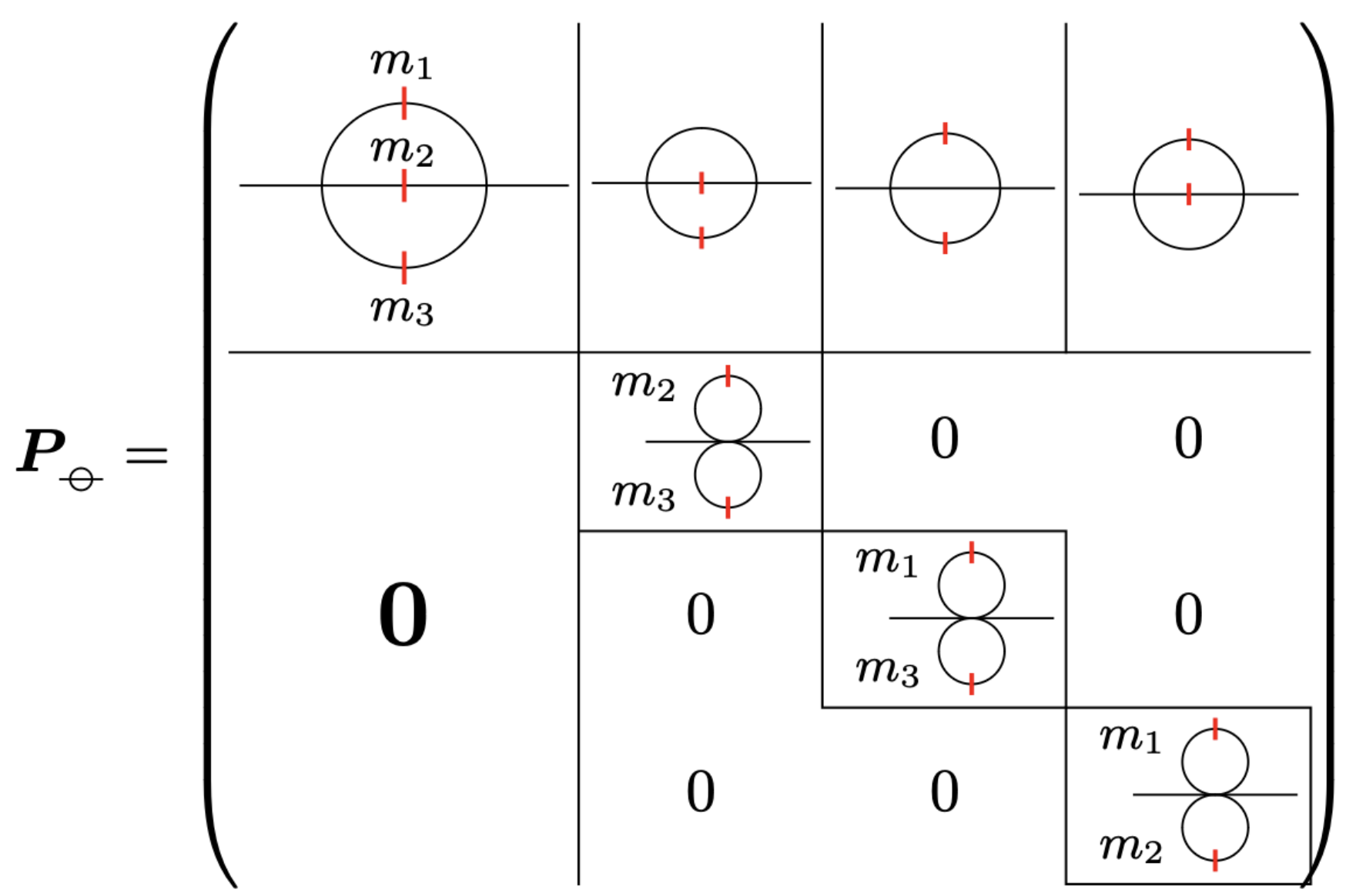}
    \label{fig:enter-label}
\end{figure}

Here the upper left block is a $4\times 4$ matrix involving the four independent maximal cut contours and master integrals of the top sector $\bTheta_1$. The other blocks in the first line have dimensions $4\times 1$ and are obtained by integrating the four master integrals of the top sector over the single maximal cut contour of the three respective subsectors $\bTheta_2,\bTheta_3,\bTheta_4$.

\end{example}
\subsection{The dual period matrix and TRBRs}
\label{sec:TRBRs-FI}

Having at hand a definition of a twisted period matrix for Feynman integrals and their cuts, we can always obtain TRBRs for a twisted period matrix of Feynman integrals. 
However, as we now proceed to show, the interpretation of TRBRs as quadratic relations among Feynman integrals and their cuts is more subtle.

If we compare the general form of the known quadratic relations in eq.~\eqref{eq:schematic_quad_rel} with the TRBRs in eq.~\eqref{generalriemann}, then at first glance the two classes of relations seem very similar. However, there is a subtle difference, which we now explain.  Equation~\eqref{eq:schematic_quad_rel}  is quadratic in the period matrix $\bP(\bx,\eps)$ (up to reversing the sign of $\eps$ in one of the terms), resulting in \emph{quadratic} relations between the entries of $\bP(\bx,\eps)$. 
%
%
The four matrices $\bC$, $\bH$, $\bP$ and $\check{\bP}$ enter \emph{linearly} into the TRBRs. Therefore, a necessary condition to obtain quadratic relations for Feynman integrals from the TRBRs is that the entries of the dual period matrix $\bs{\check{P}}$ are linear combinations of the entries of the period matrix $\bP$.

We have already seen a sufficient condition for  $\bP$ and $\check{\bP}$ to be related: if condition~\eqref{restrict} holds, one can choose bases for the twisted cohomology group and its dual according to eq.~\eqref{eq:check_to_c}. The period matrix and its dual matrix then only differ by the replacement ${\alpha}_i \rightarrow -{\alpha}_i$, cf.~eq.~\eqref{eq:non_rel_period_mat}. In that case, the TRBRs turn into relations that are quadratic in the period matrix, albeit with one of them evaluated with ${\alpha}_i \rightarrow -{\alpha}_i$. For example, if the basis of the dual cohomology group is chosen in this way, the TRBRs in eq.~\eqref{generalriemann} in the non-relative case can be cast in the form
\begin{align}
\label{epminepRiemannn}
\frac{1}{(2\pi i)^n} \bs{P}^T \left(\bs{C}^{-1}\right)^T \left({\bs{P}}|_{\alpha_i\rightarrow -\alpha_i}\right)=   \bs{H}\, ,\quad \textrm{if condition~\eqref{restrict} holds.}
\end{align}

Consider now a family of Feynman integrals, and denote the irreducible sectors by $\bTheta_1,\ldots,\bTheta_S$. With the natural order on the sectors, the matrices $\bP$ and $\bC$ are block upper triangular, while $\bs{\check{P}}$ and $\bH$ are block lower triangular. Cutting a propagator involves setting one of the propagator variables $z_i$ to zero in the Baikov polynomials. The latter are usually non-zero if we set one or more of the variables $z_i$ to zero. It follows that, for $\nu_j\in\mathbb{Z}$, the explicit factors of $z_j^{-\nu_j}$ in eq.~\eqref{eq_Baikov2} are unregulated singularities of the integrand (for $\nu_i>0$), and indeed it is easy to see that the condition in eq.~\eqref{restrict} is not satisfied. 
The only exception
are maximal cuts, because in that case we take residues in all $z_i$ with $\nu_i>0$. We conclude that for maximal cuts the TRBRs reduce to the quadratic relations in eq.~\eqref{epminepRiemannn}. We will discuss these quadratic relations for maximal cuts in detail in section~\ref{sec:MaxCutRelations}.

For non-maximal cuts, however, condition~\eqref{restrict} will not be satisfied, and so we do not expect that we can identify the dual period matrix $\bs{\check{P}}$ with the period matrix evaluated with ${\alpha_i\to-\alpha_i}$. Indeed, we have already seen in the context of Gauss' hypergeometric function in example~\ref{example21} that there is no basis in which this is possible.
Hence, if $\bP$ is the period matrix associated to a family of $L$-loop (cut) Feynman integrals, then the entries of the dual period matrix $\bs{\check{P}}$ will generically not be $L$-loop (cut) Feynman integrals from the same family, and so the TRBRs will not provide quadratic relations relating (cut) integrals from this family! 

We may nevertheless use the TRBRs in eq.~\eqref{generalriemann} to express the dual period matrix in terms of the inverse of the period matrix:
\beq\label{eq:P_dual_to_P-1}
\bs{\check{P}} = (2\pi i)^n\,\bC^T\,(\bP^{-1})^{T}\,\bH\,.
\eeq
The intersection matrices are rational in the kinematic variables $\bx$. The entries of the inverse of the period matrix have the form
\beq
(\bP^{-1})_{ji} =(\det \bP)^{-1}\,(-1)^{i+j}\,M_{ij}\,,
\eeq
where $M_{ij}$ is the minor of $\bP$ obtained by deleting the $i^{\textrm{th}}$ row and $j^{\textrm{th}}$ column of $\bP$. Since $\bP$ is block upper-triangular, the determinant of $\bP$ is given by the product of the determinants of the blocks corresponding to the maximal cuts of the different sectors
\beq
\det \bP = \prod_{i=1}^S\det \bF_{\bTheta_i}\,,
\eeq
where $\bF_{\bTheta_i}$ is the matrix of maximal cuts associated to the master integrals in the sector $\bTheta_i$. Note that, when the determinant is expanded into a Laurent series in the dimensional regulator $\eps$, it will typically involve transcendental functions of the kinematic variables $\bx$ (typically logarithms). Hence we see that, as expected, the entries of $\bs{\check{P}}$ are not $L$-loop (cut) Feynman integrals, but instead they are determinants of cuts, multiplied by the determinants associated to the maximal cuts. Thus the TRBRs enable us to solve for $\check{P}$ and write it in terms of determinants of periods (cut Feynman integrals), but don't yield further relations between them.


To conclude, we see that the TRBRs for Feynman integrals give rise to two different types of relations:
\begin{itemize}
\item The blocks on the diagonal of $\bP$ are the period matrices $\bF_{\bTheta_i}$ for the maximal cuts, and if we restrict eq.~\eqref{eq:P_dual_to_P-1} to these diagonal blocks, we obtain quadratic relations for the maximal cuts in $\bF_{\bTheta_i}$. 
\item The elements not on the diagonal blocks of $\bP$ are non-maximal cuts, and the TRBRs can be used to express the corresponding dual elements as determinants of cut integrals from this family of Feynman integrals.
\end{itemize}

Let us conclude this discussion with two comments.
First, while the previous discussion shows that it is generically not possible to obtain quadratic relations for Feynman integrals from TRBRs, there can be special cases when such quadratic relations nevertheless exist. For example, it can happen that a given sector has no lower sectors (for the natural partial ordering on sectors). In that case, all Feynman integrals and cuts from that sector are linear combinations of maximal cuts, for which the condition~\eqref{restrict} is met. Trivial examples of this include the one-loop tadpole integral and the one-loop bubble integral with massless propagators. Other, less trivial, examples are Feynman graphs obtained in the following fashion: We start from a Feynman graph without self-loops (i.e., without a banana subgraph), and then we replace each edge either by a banana graph with at least two loops and at most one massive propagator or by a one-loop massless bubble graph. It is easy to check that such a graph has no subsectors, because the contraction of any propagator leads to a graph with a detachable massless tadpole integral, which vanishes in dimensional regularisation. For these types of graphs, the Feynman integrals without cut propagators are linear combinations of maximal cuts, and the TRBRs then deliver quadratic relations. An example of such a graph is shown in figure~\ref{trianglewithbubbles}.

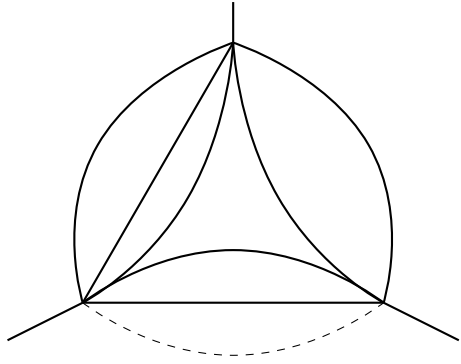
\begin{figure}
    \centering
   	\begin{tikzpicture}
 	
 	\draw[thick] plot [smooth, tension=1] coordinates {(-2,0) (0,0.7) (2,0)};
 	\draw[thick] plot [smooth, tension=1] coordinates {(2,0) (0.6,1.4) (0,3.46)};
 	\draw[thick] plot [smooth, tension=1] coordinates {(0,3.46) (-0.6,1.4) (-2,0)};
 	
 	\draw[dashed] plot [smooth, tension=1] coordinates {(-2,0) (0,-0.7) (2,0)};
 	\draw[thick] plot [smooth, tension=1] coordinates {(2,0) (1.8,2.1) (0,3.46)};
 	\draw[thick] plot [smooth, tension=1] coordinates {(0,3.46) (-1.8,2.1) (-2,0)};
	
	\draw[thick] plot [smooth] coordinates {(2,0) (-2,0)};
	\draw[thick] plot [smooth] coordinates {(0,3.46) (-2,0)};
 	
 	\draw[thick] (-3,-0.5) -- (-2,0) ;
 	\draw[thick] (3,-0.5) -- (2,0) ;
 	\draw[thick] (0,4) -- (0,3.46) ;
%
%
 	\end{tikzpicture}
\caption{A six-loop integral obtained from a one-loop triangle integral by inserting banana integrals depending on at most one mass. Solid internal lines denote massless propagators, while the dashed line represents a massive propagator. The contraction of any internal line leads to a graph with a detachable massless tadpole integral, which vanishes in dimensional regularisation.}
\label{trianglewithbubbles}
\end{figure}

Second, one may wonder if one can obtain quadratic relations for non-maximally cut integrals with integer exponents by considering deformations of the exponents that allow us to work in non-relative twisted cohomology even for non-maxmially cut Feynman integrals. To be more precise, consider an integral $\hat{I}_{\bs{\nu}}^D$ with propagator exponents $\nu_i\in\mathbb{Z}$. We can define a deformed integral 
\beq\label{eq:deformation}
\tilde{I}_{\bnu}^D(\rho) = \hat{I}_{\bnu+\rho}^D\,,\qquad \bnu+\rho = (\nu_1+\rho,\ldots,\nu_N+\rho)\,.
\eeq
Since all propagators of the deformed integral have non-integer exponents, we can choose a basis of dual co-cycles such that the (deformed) dual period matrix satisfies
\beq\label{eq:deformed_per_mat}
\bs{\check{P}}(\bx,\eps,\rho) = \bP(\bx,-\eps,-\rho)\,,\qquad \rho\notin\mathbb{Z}\,.
\eeq 
This leads to TRBRs of the form
\beq
\frac{1}{(2\pi i)^N}\bP(\bx,\eps,\rho)^T\big(\bC(\bx,\eps,\rho)^{-1}\big)^T\bP(\bx,-\eps,-\rho) = \bH(\eps,\rho)\,,
\eeq
where $\bC(\bx,\eps,\rho)$ and $\bH(\eps,\rho)$ are intersection matrices for the co-cycles and cycles computed with the deformed propagators. We obtain quadratic relations involving the entries of the deformed period matrix $\bP(\bx,\eps,\rho)$. One may be tempted to take the limit $\rho\to0$, to obtain quadratic relations between the undeformed integrals. This is similar to the approach advocated in ref.~\cite{Brunello:2023rpq} for the computation of intersection numbers between co-cycles for Feynman integrals. Indeed, while in principle the computation of intersection numbers for Feynman integrals requires the framework of relative twisted cohomology~\cite{Caron-Huot:2021xqj,Caron-Huot:2021iev}, ref.~\cite{Brunello:2023rpq} proposes to deform the exponents, compute the intersection numbers using techniques developed for the non-relative case, and then to only keep the leading terms in the limit $\rho\to 0$. As explained in detail in ref.~\cite{Brunello:2023rpq}, this approach leads to the same results as a direct computation of the intersection matrices in a relative framework in some basis. Only keeping the leading order terms in the intersection matrices requires a rotation of the dual basis. This is exactly the transformation that rotates $P(-\varepsilon)$ into $P^\text{rel}$ in the limit $\rho=0$. In that way, one can reproduce the results of relative twisted cohomology. This is sufficient for the purposes of ref.~\cite{Brunello:2023rpq}, which focuses on the question of how to decompose Feynman integral into a basis, and the entries of the dual period matrix never enter explicitly. In our case, however, the dual period matrix directly enters the TRBRs, and we expect that the relation between the period matrix and its dual in eq.~\eqref{eq:deformed_per_mat} is violated for $\rho=0$. Hence, we do not expect to obtain a quadratic relation between entries of the period matrix, i.e., between Feynman integrals, in the limit. We illustrate this by an explicit computation on the example of Gauss' hypergeometric function in appendix~\ref{app:deform}, where we show that we recover exactly the same TRBRs as in the relative case in the limit $\rho\to0$.

\begin{example}[The massive bubble in $D=2-2\eps$ dimensions] 
We illustrate the  previous discussion on the
example of the one-loop bubble integral from example~\ref{ex:bubble_period_matrix} with $m_1=m$ and $m_2=0$. 
The number of master integrals is reduced with respect to the case $m_2\neq 0$ discussed in example~\ref{ex:bubble_period_matrix}.
The period matrix for this case can be obtained from the period matrix in eq.~\eqref{eq:bubble_per_mat} by deleting the last row and column:
\beq
{\bs{P}}_{\tikz[baseline=-0.5ex]{\draw (0,0) circle [radius=2pt];
		\draw (-4pt,0) -- (-2pt,0);
		\draw (2pt,0) -- (4pt,0);}}
=
\begin{pmatrix}
	\tikz[baseline=-0ex]{
		\draw (0,0) circle [radius=5pt];
		\draw (-10pt,0) -- (-5pt,0);
		\draw (5pt,0) -- (10pt,0);
		\draw [red, thick] (0,2pt) -- (0,8pt); 
		\draw [red, thick] (0,-2pt) -- (0,-8pt); 
	}
	& \hspace{0.2cm}
	\tikz[baseline=-0ex]{
		\draw (0,0) circle [radius=5pt];
		\draw (-10pt,0) -- (-5pt,0);
		\draw (5pt,0) -- (10pt,0);
		\draw [red, thick] (0,2pt) -- (0,8pt); 
	}
	\\
	0&   \hspace{0.2cm}                                                
	\tikz[baseline=-1ex]{
		\draw (0,0) circle [radius=5pt];
		\draw (0,-5pt) -- (0,-10pt); 
		\draw [red, thick] (0,2pt) -- (0,8pt); 
	}
\end{pmatrix}\,.
\eeq
The analytic expressions for the cuts of the master integrals from example~\ref{ex:bubble_period_matrix} are (cf.~ref.~\cite{Abreu:2017mtm}):
\beq\bsp \label{bubbexp}
\tikz[baseline=-1ex]{
		\draw (0,0) circle [radius=5pt];
		\draw (0,-5pt) -- (0,-10pt); 
		\draw [red, thick] (0,2pt) -- (0,8pt); 
	}&=e^{\gamma_E\varepsilon} \frac{m^{-2\eps}\,e^{-i\pi\varepsilon}}{\Gamma(1-\varepsilon)}\,, \\
\tikz[baseline=-0.5ex]{
		\draw (0,0) circle [radius=5pt];
		\draw (-10pt,0) -- (-5pt,0);
		\draw (5pt,0) -- (10pt,0);
		\draw [red, thick] (0,2pt) -- (0,8pt); 
	}&=e^{\gamma_E\varepsilon}\frac{m^{-2\eps}\,e^{-i\pi\varepsilon}}{p^2} \frac{1}{\Gamma(1-\varepsilon)} \,{}_2F_1\left(1,1+\varepsilon; 1-\varepsilon;\frac{m^2}{p^2}\right)\,,\\
\tikz[baseline=-0.5ex]{
		\draw (0,0) circle [radius=5pt];
		\draw (-10pt,0) -- (-5pt,0);
		\draw (5pt,0) -- (10pt,0);
		\draw [red, thick] (0,2pt) -- (0,8pt); 
		\draw [red, thick] (0,-2pt) -- (0,-8pt); 
	}&= -2(p^2)^\varepsilon (p^2-m^2)^{-2\varepsilon-1} e^{\gamma_E\varepsilon} \frac{\Gamma(1-\varepsilon)}{\Gamma(1-2\varepsilon)} \,,
\esp\eeq
and these expressions are valid in the region $p^2>m^2>0$. 

Let us now discuss how to obtain the dual period matrix and the intersection matrices. We can either compute them directly (see, e.g., refs.~\cite{Caron-Huot:2021xqj,Caron-Huot:2021iev,Giroux:2022wav,Crisanti:2024onv}), or we can note that the period matrix is a special case of Gauss' hypergeometric function discussed in example~\ref{example21}. We have the relation
\begin{align}
     {\bs{P}}_{\tikz[baseline=-0.5ex]{\draw (0,0) circle [radius=2pt];
                             \draw (-4pt,0) -- (-2pt,0);
                             \draw (2pt,0) -- (4pt,0);}}= \bs{T}^c \bs{P}^\text{rel}\left(x=\tfrac{p^2}{m_1^2}, \alpha_0=\varepsilon, \alpha_1=-2\varepsilon\right) \bs{T}^h\,,
\end{align}
where $\bs{P}^\text{rel}(x,\ \balpha)$ was defined in eq.~\eqref{pmatrel}, and we defined the matrices
\beq\bsp
    \bs{T}^c&\,=\left(
\begin{array}{cc}
 \frac{m^{-2\varepsilon } \csc (\pi  \epsilon )}{(m^2-p^2) \Gamma (1-2 \epsilon ) \Gamma (\epsilon )} & 0 \\
 0 & -\frac{m^{-2\epsilon } \csc (\pi  \varepsilon )}{2 \Gamma (1-2 \varepsilon ) \Gamma (\varepsilon )} \\
\end{array}
\right)\,,\\
    \bs{T}^h&\,=\left(
\begin{array}{cc}
 2 \pi e^{2 \pi i \epsilon} & 0 \\
 0 & 2 e^{-i \pi  \varepsilon } \sin (\pi  \varepsilon ) \\
\end{array}
\right)\,.
\esp\eeq
The dual period matrix is
\beq\bsp
     {\bs{\check{P}}}_{\tikz[baseline=-0.5ex]{\draw (0,0) circle [radius=2pt];
                             \draw (-4pt,0) -- (-2pt,0);
                             \draw (2pt,0) -- (4pt,0);}}&\,= \bs{\check{T}}^c \bs{\check{P}}^{\text{rel}}\left(x=\tfrac{p^2}{m_1^2},\, \alpha_0=\varepsilon, \alpha_1=-2\varepsilon\right)\bs{\check{T}}^h\\
                             &\,=
                             \begin{pmatrix}  
\hspace{1cm} \tikz[baseline=-0ex]{
		\draw (0,0) circle [radius=5pt];
		\draw (-10pt,0) -- (-5pt,0);
		\draw (5pt,0) -- (10pt,0);
		\draw [red, thick] (0,2pt) -- (0,8pt); 
		\draw [red, thick] (0,-2pt) -- (0,-8pt); 
	}|_{\varepsilon\rightarrow -\varepsilon} &0\\ \hspace{-1cm}\hspace{1cm}\frac{\pi \csc (\pi  \varepsilon ) p^{2\epsilon }}{\Gamma(1+2\varepsilon)} \, {}_2F_1\left(-\varepsilon, -2\varepsilon; 1-\varepsilon ;\tfrac{m^2}{p^2}\right)\vspace{0.5cm}&  \tikz[baseline=-1ex]{
		\draw (0,0) circle [radius=5pt];
		\draw (0,-5pt) -- (0,-10pt); 
		\draw [red, thick] (0,2pt) -- (0,8pt); 
	}|_{\varepsilon\rightarrow -\varepsilon}                         \end{pmatrix}\,,
\esp\eeq
with
\beq\bsp
    \bs{\check{T}}^c&\,=\left(
\begin{array}{cc}
 \frac{ m^{2\varepsilon } \Gamma (\epsilon +1)}{\pi  (p^2-m^2) \Gamma (2 \varepsilon +1)} & 0 \\
 0 & -\frac{m^{2\varepsilon } \csc (\pi  \varepsilon )}{2 \Gamma (-\varepsilon ) \Gamma (2 \varepsilon +1)} 
\end{array}
\right)\,, \\
 \bs{\check{T}}^h&\,=\left(
\begin{array}{cc}
 2 \pi e^{-2 \pi i \epsilon} & 0 \\
 0 & 2 e^{i \pi  \varepsilon } \sin (\pi  \varepsilon ) \\
\end{array}
\right)\,.
\esp\eeq
The intersection matrices are
\beq\bsp
\bs{C}_{\tikz[baseline=-0.5ex]{\draw (0,0) circle [radius=2pt];
                             \draw (-4pt,0) -- (-2pt,0);
                             \draw (2pt,0) -- (4pt,0);}}&\,=\bs{T}^c\bs{C}^\text{rel}(\check{\bs{T}}^c)^T= \left(
\begin{array}{cc}
 \frac{\cos (\pi  \epsilon )}{\pi ^2 (m^2-p^2)^2} & 0 \\
 0 & \frac{\cos (\pi  \epsilon )}{2 \pi ^2 \epsilon } \\
\end{array}
\right)\,,\\ 
\bs{H}_{\tikz[baseline=-0.5ex]{\draw (0,0) circle [radius=2pt];
                             \draw (-4pt,0) -- (-2pt,0);
                             \draw (2pt,0) -- (4pt,0);}}&\,=\check{\bs{T}}^h\bs{H}^\text{rel}(\bs{T}^h)^T=\left(
\begin{array}{cc}
 -2 \pi i& 0 \\
 0 & -i \tan (\pi  \epsilon ) \\
\end{array}
\right)\,.
                             \esp\eeq
                             
We see that the entries of the dual period matrix $  {\bs{\check{P}}}_{\tikz[baseline=-0.5ex]{\draw (0,0) circle [radius=2pt];
                             \draw (-4pt,0) -- (-2pt,0);
                             \draw (2pt,0) -- (4pt,0);}}$ that are not in the diagonal blocks are not manifestly expressible in terms of Feynman integrals. We can express them in terms of Feynman integrals using eq.~\eqref{eq:P_dual_to_P-1}, and we find
\beq
\Big({\bs{\check{P}}}_{\tikz[baseline=-0.5ex]{\draw (0,0) circle [radius=2pt];
                             \draw (-4pt,0) -- (-2pt,0);
                             \draw (2pt,0) -- (4pt,0);}}\Big)_{21} =  \frac{i \cos(\pi \varepsilon)}{\pi \varepsilon}\hspace{0.2cm}
	\tfrac{\tikz[baseline=-0ex]{
		\draw (0,0) circle [radius=5pt];
		\draw (-10pt,0) -- (-5pt,0);
		\draw (5pt,0) -- (10pt,0);
		\draw [red, thick] (0,2pt) -- (0,8pt); 
	}}{\tikz[baseline=-1ex]{
		\draw (0,0) circle [radius=5pt];
		\draw (-10pt,0) -- (-5pt,0);
		\draw (5pt,0) -- (10pt,0);
		\draw [red, thick] (0,2pt) -- (0,8pt); 
		\draw [red, thick] (0,-2pt) -- (0,-8pt); 
	}\,\tikz[baseline=-1ex]{
		\draw (0,0) circle [radius=5pt];
		\draw (0,-5pt) -- (0,-10pt); 
		\draw [red, thick] (0,2pt) -- (0,8pt); 
	}}\,,
                             \eeq
where we used the fact that $\det \bs{P}_{\tikz[baseline=-0.5ex]{\draw (0,0) circle [radius=2pt];
                             \draw (-4pt,0) -- (-2pt,0);
                             \draw (2pt,0) -- (4pt,0);}} = \tikz[baseline=-1ex]{
		\draw (0,0) circle [radius=5pt];
		\draw (-10pt,0) -- (-5pt,0);
		\draw (5pt,0) -- (10pt,0);
		\draw [red, thick] (0,2pt) -- (0,8pt); 
		\draw [red, thick] (0,-2pt) -- (0,-8pt); 
	}\,\tikz[baseline=-1ex]{
		\draw (0,0) circle [radius=5pt];
		\draw (0,-5pt) -- (0,-10pt); 
		\draw [red, thick] (0,2pt) -- (0,8pt); 
	}$, and also the well-known identity
	\beq\label{eq:2F1identity}
	{}_2F_1(a,b;c;x) = (1-x)^{c-a-b}\,{}_2F_1(c-a,c-b;c;x)\,.
	\eeq
	We see that, as expected, we have
	\beq
	{\bs{\check{P}}}_{\tikz[baseline=-0.5ex]{\draw (0,0) circle [radius=2pt];
                             \draw (-4pt,0) -- (-2pt,0);
                             \draw (2pt,0) -- (4pt,0);}} \neq \bs{P}_{\tikz[baseline=-0.5ex]{\draw (0,0) circle [radius=2pt];
                             \draw (-4pt,0) -- (-2pt,0);
                             \draw (2pt,0) -- (4pt,0);}}|_{\eps\to-\eps}\,.
                             \eeq
                             
                             We can now easily write down the TRBRs for the one-loop bubble integral. We obtain three non-trivial relations. Two of them are quadratic relations for the maximal cuts:
                             \beq\bsp
\label{maxcutrel1}
\frac{4 \cos (\pi  \varepsilon )}{ (m^2-p^2)^2}&=  \tikz[baseline=-0.5ex]{
		\draw (0,0) circle [radius=5pt];
		\draw (-10pt,0) -- (-5pt,0);
		\draw (5pt,0) -- (10pt,0);
		\draw [red, thick] (0,2pt) -- (0,8pt); 
		\draw [red, thick] (0,-2pt) -- (0,-8pt); 
	}\left(\tikz[baseline=-0.5ex]{
		\draw (0,0) circle [radius=5pt];
		\draw (-10pt,0) -- (-5pt,0);
		\draw (5pt,0) -- (10pt,0);
		\draw [red, thick] (0,2pt) -- (0,8pt); 
		\draw [red, thick] (0,-2pt) -- (0,-8pt); 
	}|_{\varepsilon\rightarrow-\varepsilon}\right)\,,\\
\frac{\sin (\pi  \varepsilon )}{\pi  \varepsilon }&= \tikz[baseline=-0.5ex]{
		\draw (0,0) circle [radius=5pt];
		\draw (0,-5pt) -- (0,-10pt); 
		\draw [red, thick] (0,2pt) -- (0,8pt); 
	} \left(\tikz[baseline=-0.5ex]{
		\draw (0,0) circle [radius=5pt];
		\draw (0,-5pt) -- (0,-10pt); 
		\draw [red, thick] (0,2pt) -- (0,8pt); 
	}|_{\varepsilon\rightarrow-\varepsilon}\right).
\esp\eeq
The third relation is simply the well-known hypergeometric identity in eq.~\eqref{eq:2F1identity}, which is the relation that was needed to express the off-diagonal elements of ${\bs{\check{P}}}_{\tikz[baseline=-0.5ex]{\draw (0,0) circle [radius=2pt];
                             \draw (-4pt,0) -- (-2pt,0);
                             \draw (2pt,0) -- (4pt,0);}}$ as Feynman integrals using eq.~\eqref{eq:P_dual_to_P-1}. Hence, as expected, we see that the TRBRs factorise into two groups: one which allows us to write the off-diagonal elements of the dual period matrix in terms of Feynman integrals, and the other which give quadratic relations between maximal cuts evaluated at $\eps$ and $-\eps$ respectively.

\end{example}

\section{Quadratic relations for maximal cuts}
\label{sec:MaxCutRelations}

\subsection{Quadratic relations maximal cuts from TRBRs}

In the previous section, we have argued that TRBRs lead to quadratic relations among maximal cuts. In this section, we study those relations in detail, and we show how they are connected to the quadratic relations from the literature (see section~\ref{sec.intro}).


Consider a family of Feynman integrals and a sector $\bTheta$, which we assume without loss of generality to be of the form
\beq
\bTheta = (\underbrace{1,\ldots,1}_{m},\underbrace{0,\ldots,0}_{N-m})\,.
\eeq
In this sector, we may pick a basis of master integrals of the form
\beq
\varphi_i = \frac{\rd z_1}{z_1}\wedge\ldots\wedge\frac{\rd z_m}{z_m}\wedge\psi_i\,, \qquad \psi_i = \rd^{h}z\,f_i(\bz)\,,\qquad i=1,\ldots,M_{\bTheta}\,,
\eeq
where $h=N-m$ and  $f(\bz)$ is a polynomial. A basis of co-cycles for the maximal cuts is then 
\beq
\psi_i|_{z_1=\ldots=z_m=0} = \rd^{h}z\,f^{\textrm{m.c.}}_i(\bz)\,,\qquad i=1,\ldots,M_{\bTheta}\,,
\eeq
where for a polynomial $f(z_1,\ldots,z_N)$ we define $f^{\textrm{m.c.}}(z_{m+1},\ldots,z_N) = f(0,\ldots,0,z_{m+1},\ldots,z_N)$.
If we fix a basis of maximal cut cycles $\mathcal{C}_j$, $j=1,\ldots,M_{\bTheta}$, then the period matrix for the maximal cuts has entries
\beq\bsp\label{eq:MaxCut_per_mat}
F_{ij}(\bx,\eps)&\, = \int_{\mathcal{C}_j}\left[\psi_i\prod_{k=1}^K\mathcal{B}_k(\bz)^{\mu_k}\right]_{{z_1=\ldots=z_m=0}} = \int_{\mathcal{C}_j}\rd^{h}z\,f_i^{\textrm{m.c.}}(\bz)\,\prod_{k=1}^K\left(\mathcal{B}_k^{\textrm{m.c.}}(\bz)\right)^{\frac{m_k}{2}+\sigma_k\eps}\,,
\esp\eeq
with $\mu_i$ defined in eq.~\eqref{eq:Baikov_twis_exponent}.


We can define a basis of dual co-cycles using eq.~\eqref{eq:check_to_c}. The dual period matrix is given by eq.~\eqref{eq:checkP_to_minus}, which now takes the form
\beq
\bs{\check{F}}(\bx,\eps) = \bF(\bx,\eps)|_{\mu_i\to -\mu_i} \neq \bF(\bx,-\eps)\,.
\eeq
It is possible to find a dual basis such that the condition $\bs{\check{F}}(\bx,\eps) = \bs{{F}}(\bx,-\eps)$ holds. Indeed, if we pick the dual basis
\beq\label{eq:FI_dual_basis}
\tilde{\psi}_i = \left[\psi_i\prod_{k=1}^K\left(\mathcal{B}_k^{\textrm{m.c.}}(\bz)\right)^{m_k}\right]_c\,,
\eeq
then we find\footnote{We suppress overall proportionality factors for readibility.}
\beq\bsp\label{eq:dual_-eps}
\check{F}_{ij}(\bx,\eps)  &\,= \int_{\mathcal{C}_j}\psi_i \,\prod_{k=1}^K\left(\mathcal{B}_k^{\textrm{m.c.}}(\bz)\right)^{m_k}\,\left[\prod_{k=1}^K\left(\mathcal{B}^{\textrm{m.c.}}_k(\bz)\right)^{\frac{m_k}{2}+\sigma_k\eps}\right]^{-1}\\
&\, = \int_{\mathcal{C}_j}\psi_i \prod_{k=1}^K\left(\mathcal{B}^{\textrm{m.c.}}_k(\bz)\right)^{\frac{m_k}{2}-\sigma_k\eps}\\
&\,= {F}_{ij}(\bx,-\eps)\,.
\esp\eeq

The TRBRs in eq.~\eqref{generalriemann} then imply that we have the quadratic relations
\beq\label{eq:MaxCuts_TRBR_H}
\frac{1}{(2\pi i)^h}\,\bF(\bx,\eps)^T\big(\bC(\bx,\eps)^{-1}\big)^T\bF(\bx,-\eps) = \bH(\eps)\,,
\eeq
or equivalently
\beq\label{eq:MaxCuts_TRBR_C}
\frac{1}{(2\pi i)^h}\,\bF(\bx,\eps)\big(\bH(\eps)^{-1}\big)^T\bF(\bx,-\eps)^T = \bC(\bx,\eps)\,,
\eeq
where $\bH(\eps)$ and $\bC(\bx,\eps)$ are the intersection matrices for twisted cycles and co-cycles. These quadratic relations have the form anticipated in eq.~\eqref{eq:schematic_quad_rel}. 
In the remainder of this section, we will show that the quadratic relations for maximal cuts from CY geometry at $\eps=0$~\cite{Bonisch:2021yfw} and for integrals depending on a single dimensionless variable in dimensional regularisation~\cite{Lee:2018jsw} are special cases of eqs.~\eqref{eq:MaxCuts_TRBR_H} and~\eqref{eq:MaxCuts_TRBR_C}. 


\subsection{Relationship to the quadratic relations from CY geometry}
\label{subsec:quadraticsCY}

Let us start by discussing how the quadratic relations in eq.~\eqref{eq:MaxCuts_TRBR_C} reduce to the quadratic relations for maximal cuts that evaluate to (quasi-)periods of CY geometries that were obtained in \cite{Duhr:2022dxb} using methods from CY geometry. The main difference between eqs.~\eqref{eq:MaxCuts_TRBR_C} and the relations of \cite{Duhr:2022dxb} is that the latter are derived in integer dimensions, where $\eps=0$, without any reference to twisted cohomology theories.  We now argue that exactly the same relations can be obtained from the TRBRs for maximal cuts. 

From eq.~\eqref{eq:MaxCut_per_mat} and~\eqref{eq:dual_-eps} we see that for $\eps=0$, the exponents of the Baikov polynomials are integers or half-integers. In the applications to Feynman integrals associated to CY geometries, the exponents can be chosen to be half integers.\footnote{This may require to integrate out some Baikov parameters first, or equivalently, to start from a specific loop-by-loop representation.} More precisely, in those cases we have
\beq\label{eq:CY_periods}
F_{ij}(\bx,\eps=0) = \int_{\mathcal{C}_j}\frac{\rd^{h}z}{\sqrt{P(\bz)}}\,f_i^{\textrm{m.c.}}(\bz)\,,
\eeq
where $P(\bz)$ is a polynomial. Then the family of varieties defined by the equation $y^2=P(\bz)$ is a family of CY varieties, with holomorphic differential $\Omega=\frac{\rd^{h}z}{\sqrt{P(\bz)}}$. The maximal cuts in eq.~\eqref{eq:CY_periods} can then be identified with the periods and quasi periods, and $\bF(\bx,\eps=0)$ is the period matrix of this family. Quadratic relations for the period matrix then follow from the Hodge structure on the middle cohomology of a family of CY varieties and Griffith transversality for the Gauss-Manin connection describing the variation of the Hodge structure. They take the form~\cite{MR717607}
\beq\label{eq:CY_Griffith}
\bF(\bx)\bSigma \bF(\bx)^T = \bZ(\bx)\,,
\eeq
where $\bZ(\bx)$ is a matrix of rational functions of the kinematic variables (which correspond to the independent moduli of the family of CY varieties), and $\Sigma$ is (the inverse of) the matrix of intersection numbers between generators of the middle homology of the CY variety. These quadratic relations are a direct consequence of the geometric interpretation of the maximal cuts as (quasi-)periods of a family of CY varieties, and their relevance for maximal cuts in integer dimensions was first pointed out in ref.~\cite{Duhr:2022dxb}. Note that the period matrix is the fundamental solution matrix of the Gauss-Manin connection describing the family.

Alternatively, we may view the integrals in eq.~\eqref{eq:CY_periods} as periods of a twisted cohomology theory, even for $\eps=0$. The twist is given by
\beq\label{eq:Phi_CY}
\Phi_{\textrm{CY}} = P(\bz)^{-\frac{1}{2}}\,,
\eeq
and a basis of twisted cycles and co-cycles is $\mathcal{C}_j$ and $\psi_i=\rd^{h}z\,f_i^{\textrm{m.c.}}(\bz)$. Our basis of dual co-cycles is $\tilde{\psi}_i = \big[\rd^{h}z\,f_i^{\textrm{m.c.}}(\bz)\,P^{-1}(\bz)\big]_c$. When written in this basis, the TRBRs agree with the quadratic relations from CY geometry. Hence, the quadratic relations among maximal cuts for Feynman integrals in integer dimensions associated to CY geometries are a special case of the quadratic relations from TRBRs. A detailed discussion about this relation between the TRBRs for a certain hypergeometric function and the Riemann (in-)equality of the related K3 surface can also be found in ref.~\cite{yoshida_hypergeometric_1997}. 

Let us conclude by making a comment of how these relations may extend to Feynman integrals in integer dimensions associated to geometries that are not CY. In fact, in our discussion we only relied on the fact that the integrand in eq.~\eqref{eq:CY_periods} contains a square root, which allows us to define the twist in eq.~\eqref{eq:Phi_CY}, but it did not use any property specific to CY geometries. In particular, the same discussion can be applied to other geometries defined as a double-cover given by an equation of the form $y^2=P(\bz)$, irrespective if they are CY or not. This includes in particular the case of families of hyperelliptic curves, with the maximal cuts computing the periods and quasi-periods of those families of hyperelliptic curves. The periods of Riemann surfaces are well known to satisfy quadratic relations which are the classical Riemann bilinear relations. The quadratic relations from TRBRs between maximal cuts in integer dimensions of Feynman integrals associated to higher-genus hyperelliptic curves (like those considered in refs.~\cite{Huang:2013kh,Hauenstein:2014mda,Marzucca:2023gto}) will therefore agree with the classical Riemann bilinear relations between the periods of those curves. We will explicitly show this in the example of the non-planar crossed box in section \ref{subsec:npBox}


\subsection{Relationship to the quadratic relations for one-variable integrals}
Let us now establish the relationship between the quadratic relations for maximal cuts from TRBRs in eqs.~\eqref{eq:MaxCuts_TRBR_H} and those for maximal cuts depending on a single dimensionless ratio $x$ from ref.~\cite{Lee:2018jsw}. 

The starting point of ref.~\cite{Lee:2018jsw} is the differential equation satisfied by the matrix $\bF(x,\eps)$ of maximal cuts. 
More precisely, ref.~\cite{Lee:2018jsw} conjectures that there is a basis such that
\beq\label{eq:Lee_conjecture}
\rdex \bF(x,\eps) = \mu\,\bS(x)\bF(x,\eps)\,,\qquad \mu\in\{\eps,\tfrac{1}{2}+\eps\}\,,
\eeq
where $\bS(x)^T=\bS(x)$ is a symmetric matrix. From this conjecture follows that $\bF(x,\eps)$ satisfies a set of quadratic relations: 
\begin{itemize}
\item if $\mu=\eps$, then the fundamental solution takes the form of a path-ordered exponential,
\beq\label{eq:Pexp}
\bF(x,\eps) = \mathbb{P}\exp\left[\eps\int_{\gamma}\bS(x)\right]\,,
\eeq
and from the symmetry of $\bS(x)$ it follows that
\beq\label{quad1}
\bF(x,-\eps)^T\bF(x,\eps) = \mathds{1}\,.
\eeq
\item if $\mu=\tfrac{1}{2}+\eps$, then we have the quadratic relation:
\begin{align}
\label{quad2}
    \bs{F}(x,-\varepsilon)^T  \bs{R}(x,\varepsilon)  \bs{F}(x,\varepsilon) = \bs{\widetilde{H}}(\varepsilon)\, , 
\end{align}
where $\bs{\widetilde{H}}(\varepsilon)$ is independent of $x$, and $\bs{R}(x,\varepsilon)$ is the dimension-shift matrix from eq.~\eqref{eq:R_dim_shift}.
\end{itemize}


We now argue that the conjectures about the differential equation and the ensuing quadratic relations in eqs.~\eqref{quad1} and~\eqref{quad2} follow naturally from the framework of twisted cohomology. At the same time, we will see that the derivation from twisted cohomology is independent of the number of kinematic variables, establishing that the results of ref.~\cite{Lee:2018jsw} hold more broadly than just for integrals depending on a single variable.

\subsubsection{The differential equation for maximal cuts} 
We start by discussing how one can derive eq.~\eqref{eq:Lee_conjecture} from the framework of twisted cohomology. As a starting point, let us assume that the period matrix $\bF(\bx,\eps)$ for the maximal cuts satisfies the differential equation (cf.~eq.~\eqref{eq:Gauss-Manin})
\beq
\rdex \bF(\bx,\eps) = \bs{\Omega}(\bx,\eps)\,\bF(\bx,\eps)\,,
\eeq
where $\bs{\Omega}(\bx,\eps)$ is a matrix whose entries are rational one-forms in $\bx$ and rational functions in $\eps$. Before we proceed, we make an assumption: a typical basis of twisted cycles are the chambers in $\mathbb{R}^h$ bounded by zeroes of the twist, i.e., the chambers in $\mathbb{R}^h- \Sigma$, with
\beq
\Sigma=\bigcup_{k=1}^K\big\{\bz\in\mathbb{R}^h : \mathcal{B}_k^{\textrm{m.c.}}(\bz)=0\big\}\,.
\eeq
Likewise, a typical basis of twisted co-cycles then has logarithmic singularities on the boundaries of those chambers, i.e., with logarithmic singularities along $\Sigma$. While such a \emph{logarithmic basis} is expected to exist quite generally, explicit constructions are only known in cases where $\Sigma$ is a union of hyperplanes and at most one hypersurface is defined by a polynomial of degree $h>1$~\cite{Mastrolia:2018uzb}. For instance, if $\Sigma$ is a union of linear hyperplanes only, then it is possible to choose a $\rd\!\log$ basis of the form 
\begin{align}
\label{logbasis}
\varphi_I =\rd \!\log \left(\frac{L_{i_0}}{L_{i_1}}\right)\wedge \rd \!\log\left(\frac{L_{i_1}}{L_{i_2}}\right) \wedge \dots \wedge \rd\! \log \left(\frac{L_{i_{h-1}}}{L_{i_h}}\right) \, , 
\end{align}
where $I=(i_0,i_1,\dots, i_h)$~\cite{aomoto_theory_2011}. While the construction of a logarithmic basis is still an open question in the general case, the results of ref.~\cite{Mastrolia:2018uzb} cover already many interesting cases of maximal cuts, cf.~e.g., ref.~\cite{Bosma:2017ens}. In particular, the examples considered in ref.~\cite{Lee:2018jsw} are  expressible in terms of logarithmic bases.
We will now show the following result:

\begin{theorem}
\label{alphaformtheorem}Consider a twisted cohomology group with twist
\beq \label{formtwist}
\Phi=\prod_{k=1}^L\mathcal{B}_k(\bz)^{a_k\mu}\,, \qquad a_k\in\mathbb{Q}\,,
\eeq
where $\mu$ is a formal variable,
and assume that the period matrix $\bF(\bx,\mu)$ satisfies the differential equation 
\beq
\rdex\bF(\bx,\mu) = \bs{\Omega}(\bx,\mu) \bF(\bx,\mu)\,.
\eeq
If $\bF(\bx,\mu) $ is the period matrix for a choice of a logarithmic basis of twisted co-cyles, then the matrix entering the differential equation takes the form
\beq\label{eq:Omega_mu_S}
\bs{\Omega}(\bx,\mu) = \mu\,\bs{\tilde{\Omega}}(\bx)\,.
\eeq
Moreover, it is possible to pick a logarithmic basis of twisted co-cycles such that $\bs{\tilde{\Omega}}(\bx)=\bs{\tilde{\Omega}}(\bx)^T$ is a symmetric matrix. 
\begin{proof}
Let us fix a basis of twisted cycles $\gamma_j$ and a logarithmic basis of twisted co-cycles $\varphi_i$, $1\le i,j\le M$. The basis of dual co-cycles $\check{\varphi}_i$ can be chosen as in eq.~\eqref{eq:check_to_c}.
The differential equation is:
\beq\bsp
    \rd_{\text{ext}}{F}_{ij}(\bs{x},\mu)&= \rd_{\text{ext}} \langle \varphi_i|\gamma_j]= \langle \rd_{\text{ext}}\varphi_i +\rd_\text{ext} \log \Phi\wedge\varphi_i|\gamma_j]\\
    &= \frac{1}{(2\pi i)^h} \sum_{l,k} \langle \eta_i |\check{\varphi}_{l}\rangle \big(\bC(\bx,\mu)^{-1}\big)_{lk} F_{kj}(\bx,\mu)\,,
    \esp\eeq
     with $C_{ij}(\bx,\mu) = \frac{1}{(2\pi i)^h}\langle \varphi_i|\check{\varphi}_j\rangle$, and we defined
    \beq
    \eta_i = \rd_\text{ext}\varphi_i +\rd_\text{ext} \log \Phi\wedge\varphi_i\, .
\eeq
It follows that the matrix $\bs{\Omega}(\bx,\mu)$ can be cast in terms of intersection numbers as 
\begin{align}
\label{Omegaikwithint}
    \Omega_{ik}(\bx,\mu) = \frac{1}{(2\pi i)^h}\sum_l\langle \eta_i |\check{\varphi}_{l}\rangle \big(\bC(\bx,{\mu})^{-1}\big)_{lk}\, . 
\end{align}
Since intersection numbers are rational functions, the entries of $\bs{\Omega}(\bx,{\mu})$ are rational functions of ${\mu}$.

Since we are working in a logarithmic basis, it follows from a theorem of ref.~\cite{ojm_1200788347} (see  Theorem~\ref{matsumototheorem} in appendix \ref{app:proof}) that 
\beq\label{eq:Matsumoto_thm}
C_{kl}(\bx,\mu) = \frac{1}{(2\pi i)^h}\braket{\varphi_k|\check{\varphi}_l} =  \frac{1}{\mu^h}\,a_{kl}\,,
\eeq
for some rational numbers $a_{kl}$.
Equivalently 
\beq\label{eq:Matsumoto_thm_2}
\bC(\bx,\mu) = \frac{1}{\mu^{h}}\,{\bs{\widetilde{C}}}\,,
\eeq
where $\bs{\widetilde{C}}$ is a constant matrix.
 In appendix \ref{app:proof} we prove that the $\mu$-dependence of the remaining intersection numbers is fully captured by 
\beq\label{eq:eta_phi_scale}
\langle \eta_i |\check\varphi_{l}\rangle =  \frac{(2\pi i)^h}{\mu^{h-1}}\,B_{il}(\bx)\,.
 \eeq
  Thus, we find that the entries of $\bs{\Omega}(\bx,{\mu})$ are linear in $\mu$,
\begin{align}
    \bs{\Omega}(\bx,\mu) =\mu\,\bs{B}(\bx)\bs{\widetilde{C}}^{-1}\, .
\end{align}

Let us now show that $\bs\Omega(\bx,{\mu})$ is symmetric. We start by noting that we can always choose a logarithmic basis $\varphi_i$ such that the matrix $\bC(\bx,\mu)$ of cohomology intersection numbers is diagonal. Indeed, if $\bC(\bx,\mu)$ is not diagonal in the basis $\varphi_i$, then we can change basis to an orthogonal basis $\tilde{\varphi}$ via a Gram-Schmidt procedure\footnote{Gram-Schmidt can be applied here since the cohomology intersection pairing is symmetric and non-degenerate for our choice of (dual) basis.}:
\begin{align}
\label{eq:Gram-Schmidt}
 \tilde{\varphi}_i&=\varphi_i-\sum_{j=1}^{i-1}\frac{\braket{\tilde{\varphi}_j|\check{\varphi}_i}}{\braket{\tilde{\varphi}_j|\check{\tilde{\varphi}}_j}} \tilde{\varphi}_j\,,\qquad i =1,\ldots,M\,,
\end{align}
and we have $\braket{\tilde{\varphi}_i|\check{\tilde{\varphi}}_j}=0$ for $i\neq j$. Let us write the change of basis from $\varphi_i$ to $\tilde\varphi_i$ as 
\beq
\tilde{\varphi}_i =\sum_{j=1}^{M} A_{ij}\varphi_j\, ,
\eeq
transforming the dual basis in the same way. A priori, the matrix $\bs{A}$ depends on $\bx$ and $\mu$. However, from eq.~\eqref{eq:Gram-Schmidt} we know (inductively) that the entries of $\bs{A}$ are built out of ratios of intersection numbers of the form $\braket{\varphi_i|\check{\varphi}_j}$. It follows from eq.~\eqref{eq:Matsumoto_thm} that these ratios, and therefore the entries of the matrix $\bs{A}$, are constant in both $\mu$ and $\bx$. Hence, the basis is a linear combination with constant coefficients of logarithmic forms. The matrix of intersection numbers in this basis is 
\beq
\bA\bC(\bx,\eps)\bA^T = {\mu^{-h}}\,\bA\bs{\widetilde{C}}\bA^T ={\mu^{-h}}\,\diag(\chi_1,\ldots,\chi_M)\,,
\eeq
with the constants $\chi_i$ given by
\beq
\chi_i = (2\pi i)^{-h}\,\mu^h\,\braket{\tilde\varphi_i|\check{\tilde{\varphi}}_i} \,.
\eeq
The period matrix in the basis $\tilde{\varphi}_i$ is $\bs{\widetilde{P}}(\bx,\mu) = \bs{A}\bP(\bx,\mu)$, and by construction the dual period matrix is $\bs{\check{\widetilde{P}}}(\bx,\mu) =  \bs{\widetilde{P}}(\bx,-\mu)$. Since $\bA$ is constant, the differential equation for $\bs{\widetilde{P}}(\bx,\mu)$ is still in $\mu$-factorised form,
\begin{align}
\rdex\bs{\widetilde{P}}(\bx,\mu)=\mu \bA \bs{\widetilde\Omega}(\bx)\bA^{-1}
\bs{\widetilde{P}}(\bx,\mu)\,.
\end{align}
At this point we can even change basis to an orthonormal basis, e.g., by definining
\beq\bsp
\psi_i = \sqrt{(2\pi i)^{-h}\,\mu^h\,\chi_i^{-1}}\,\tilde{\varphi_i}\textrm{~~~and~~~} \check{\psi}_i = \sqrt{(2\pi i)^{-h}\,\mu^h\,\chi_i^{-1}}\check{\tilde{\varphi}}_i\,.
\esp\eeq
It is easy to see that in this basis the intersection matrix is unity, $\braket{\psi_i|\check{\psi}_j} = \delta_{ij}$.
Thus, we can find bases for the twisted cohomology group and its dual that are at the same time orthonormal and have a differential equation in $\mu$-factorised form. We therefore assume from now on without loss of generality that our original bases $\varphi_i$ and $\check{\varphi}_i$ have this property.

We can now obtain the differential equation for the dual period matrix $\bs{\check{{P}}}$ in two ways. On the one hand, it can be obtained by replacing $\mu$ by $-\mu$ in the differential equation for the period matrix $\bs{{P}}$: 
\begin{align}
\label{eq:eqproof1}
 \rdex\bs{\check{{P}}}(\bx,\mu) =  \rdex\bs{{{P}}}(\bx,-\mu) = -\mu{\bs{\widetilde\Omega}}(\bx)\bs{{{P}}}(\bx,-\mu) = -\mu{\bs{\widetilde\Omega}}(\bx)\bs{\check{{P}}}(\bx,\mu)\,.
 \end{align}
%
On the other hand, since the matrix of intersection numbers is unity, eq.~\eqref{eq:dC} implies 
\beq\label{eq:eqproof2}
 \rdex\bs{\check{P}}(\bx,\mu) = - \bs{\Omega}(\bx,\mu)^T\bs{\check{{P}}}(\bx,\mu) = -\mu\bs{\widetilde\Omega}(\bx)^T \bs{\check{{P}}}(\bx,\mu).
 \eeq
 Comparing eqs.~\eqref{eq:eqproof1} and~\eqref{eq:eqproof2}, we see that $\bs{\widetilde\Omega}(\bx)^T = \bs{\widetilde\Omega}(\bx)$.
\end{proof}
\end{theorem}

Let us now discuss how Theorem~\ref{alphaformtheorem} implies the conjecture in eq.~\eqref{eq:Lee_conjecture}. We start from the Baikov representation in eq.~\eqref{eq_Baikov}, where we have a single Baikov polynomial. Note that if we had worked in a loop-by-loop approach, we would have more Baikov polynomials. The differential equations, however, are independent of the representation used for the integrals, so that every conclusion reached using one representation also holds for the others.
Next, we can choose a basis of twisted co-cycles such that the entries of the period matrix are
\beq\bsp\label{eq:MaxCut_per_mat}
F_{ij}(\bx,\eps)&\, = \int_{\mathcal{C}_j}\rd^{h}z\,f_i^{\textrm{m.c.}}(\bz)\,\left(\mathcal{B}^{\textrm{m.c.}}(\bz)\right)^{-\mu}\,,
\esp\eeq
where $\mu\in\{\eps,\tfrac{1}{2}+\eps\}$. It is easy to see that such a basis always exists. We now interpret these integrals as periods of some twisted cohomology theory with twist
\beq
\Phi = \left(\mathcal{B}^{\textrm{m.c.}}(\bz)\right)^{-\mu}\,.
\eeq
We can then choose a logarithmic basis of co-cycles. Let us denote the period matrix in that basis by $\bF_{\textrm{log}}(\bx,\eps)$. In that basis the hypotheses of the theorem are satisfied, and so we have
\beq
\rdex\bF_{\textrm{log}}(\bx,\eps) = -\mu\,\bs{\tilde{\Omega}}(\bx)\bF_{\textrm{log}}(\bx,\eps)\,,  \textrm{~~~~with~~~~} \bs{\tilde{\Omega}}(\bx)^T=\bs{\tilde{\Omega}}(\bx)\,.
\eeq
Hence, if we pick $\bS(\bx) = -\bs{\tilde{\Omega}}(\bx)$ we see that there is a basis such that the conjecture in eq.~\eqref{eq:Lee_conjecture} is satisfied. We see that the conjecture of ref.~\cite{Lee:2018jsw} about the structure of the differential equations for the maximal cuts naturally follows from the framework of twisted cohomology and the existence of a logarithmic basis. At the same time, while the results of ref.~\cite{Lee:2018jsw} are restricted to cases where the integral depends on a single dimensionless variable, we see that the results naturally extend to multi-variable cases.

\subsubsection{Quadratic relations for maximal cuts} \label{quadrelmaxcut}

Having established that the conjectures of ref.~\cite{Lee:2018jsw} naturally follow from twisted cohomology, let us now explain how they give rise to quadratic relations for maximal cuts. The argument follows the same lines as in ref.~\cite{Lee:2018jsw}, but we repeat the arguments here, to show how they naturally arise from twisted cohomology, and to show that they are not restricted to one-variable cases. %
From what we just discussed, we know that there is a (logarithmic) basis in which the period matrix $\bF_{\textrm{log}}(\bx,\eps)$ satisfies the differential equation
\beq\label{eq:MC_DEQ_mu}
\rdex\bF_{\textrm{log}}(\bx,\eps) =\mu\,\bS(\bx)\bF_{\textrm{log}}(\bx,\eps)\,,\qquad \mu\in\{\eps,\tfrac{1}{2}+\eps\}\,,\qquad \bS(\bx)^T=\bS(\bx)\,.
\eeq
We now discuss the two possible choices for $\mu$ in turn.

\paragraph{Quadratic relations for $\bm{\mu=\eps}$.} For $\mu=\eps$, a fundamental solution matrix of eq.~\eqref{eq:MC_DEQ_mu} can be chosen as the path-ordered exponential (cf.~eq.~\eqref{eq:Pexp}):
\beq\label{eq:MC_Pexp}
 \bU(\bx,\eps) = \mathbb{P}\exp\left[\eps\int_{\bx_0}^{\bx}{\bS}(\bx')\right] = \big(\bU(\bx,-\eps)^T\big)^{-1}\,,
 \eeq 
 where the last step follows from the symmetry of $\bS(\bx)$~\cite{Lee:2018jsw}. Note that, if $\bx_0$ coincides with one of the logarithmic singularities in the integrand, then we use a tangential base-point regularisation prescription, cf.,~e.g.,~refs.~\cite{DeligneTangential,Brown:mmv}. The path-ordered exponential in eq.~\eqref{eq:MC_Pexp} satisfies the obvious quadratic relation:
 \beq
  \bU(\bx,\eps)^T \bU(\bx,-\eps)  =\mathds{1}\,.
  \eeq
  The matrix of maximal cuts in the logarithmic basis is related to the path-ordered exponential via
  \beq
  \bF_{\textrm{log}}(\bx,\eps) = \bU(\bx,\eps)\,\bF_0(\eps)\,,
  \eeq
  where $\bF_0(\eps)$ is independent of $\bx$. Hence, we have the the quadratic relation
  \beq
  \bF_{\textrm{log}}(\bx,\eps)^T\bF_{\textrm{log}}(\bx,-\eps) = \bF_0(\eps)^T\bF_0(-\eps)\,.
  \eeq
 We see that we recover a quadratic relation of the type~\eqref{eq:MaxCuts_TRBR_H} with $\bC(\bx,\eps) = (2\pi i)^{-h}\mathds{1}$ and $\bH(\eps) = \bF_0(\eps)^T\bF_0(-\eps)$. In particular, we see that $\bH(\eps)$ is independent of $\bx$. The fact that the intersection matrix between co-cycles is unity could have been anticipated, because it follows from the proof of Theorem~\ref{alphaformtheorem}.
 We note that, while we relied on the specific basis to derive this specific form of the quadratic relation, we can obtain a quadratic relation in any other basis. Indeed, consider another basis of twisted co-cycles (but with the same basis of twisted cycles for simplicity), whose period matrix $\bF(\bx,\eps)$ is related to $\bF_{\textrm{log}}(\bx,\eps)$ via 
 \beq
 \bF_{\textrm{log}}(\bx,\eps) = \bT(\bx,\eps)\bF(\bx,\eps)\,.
 \eeq
 Then $\bF(\bx,\eps)$ satisfies the quadratic relation
 \beq
 \frac{1}{(2\pi i)^h}\,\bF(\bx,\eps)^T\big(\bC(\bx,\eps)^{-1}\big)^T\bF(\bx,-\eps) =\bF_0(\eps)^T\bF_0(-\eps) =  \bH(\eps)\,,
 \eeq
 with $\bC(\bx,\eps) = (2\pi i)^{-h}\,\bT(\bx,\eps)^{-1}\big(\bT(\bx,-\eps)^{-1}\big)^T$ .

\paragraph{Quadratic relations for $\bs{\mu}=\bs{\varepsilon+\frac{1}{2}}$.}

We choose some bases  ${\gamma_i}$ and ${\check{\gamma_i}}$, $1\le i\le M$,  for the contours and their duals. The homology intersection matrix does not depend on external parameters, but only on $\varepsilon$. We can also pick an orthonormal basis $\varphi_i$ and $\check{\varphi}_i$ for the twisted co-cycles and their duals. In those bases, the fundamental solution matrix in $D=d-2\eps$ dimensions is 
\beq
\bF_{\log}(\bx,\eps) = \bP(\bx,\mu) = \bP\big(\bx,\tfrac{1}{2}+\eps)\,.
\eeq
Then $\bF_{\log}(\bx,\eps-1)$ is the fundamental solution matrix of maximal cuts computed in $D+2$ dimensions, and it is related to $\bF_{\log}(\bx,\eps)$ via the dimension-shift matrix (cf.~eq.~\eqref{eq:R_dim_shift}),
\beq\label{eq:proof_1/2_1}
\bF_{\log}(\bx,\eps-1) = \bR(\bx,\eps)\bF_{\log}(\bx,\eps)\,.
\eeq
On the other hand, since the exponents of the twist are $\mu=\tfrac{1}{2}+\eps$, it is easy to see that
\beq
F_{\log,ij}(\bx,\eps-1) = \langle \varphi_i^+ |\gamma_j] \,,
\eeq
with $\varphi_i^+ = \varphi_i\,\cB^{\textrm{m.c.}}(\bx)$. Since the intersection matrix is $\bC(\bx,\eps) = (2\pi i)^{-h}\mathds{1}$, we have
\beq\bsp\label{eq:proof_1/2_2}
F_{\log,ij}(\bx,\eps-1) 
&\,= \sum_{k=1}^M\braket{\varphi_i ^+|\check{\varphi}_k}\,\langle \varphi_k|\gamma_j]  = \sum_{k=1}^M\braket{\varphi_i ^+|\check{\varphi}_k}\,F_{\log,kj}(\bx,\eps)\,,
\esp\eeq
Comparing eq.~\eqref{eq:proof_1/2_2} to eq.~\eqref{eq:proof_1/2_1}, we see that the dimension shift matrix can be written as a matrix of intersection numbers:
\beq\label{eq:R_intersection}
{R}_{ij}(\bx,\eps)=\braket{\varphi_i ^+|\check{\varphi}_j}\,.
\eeq

We know that we can pick a basis of dual co-cycles such that $\bs{\check{P}}(\bx,\mu) = \bP(\bx,-\mu)$. This gives
\beq\bsp
\bF_{\log}(\bx,-\eps) &\,= \bR(\bx,-\eps)^{-1}\bF_{\log}(\bx,-\eps-1)\\
&\, =\bR(\bx,-\eps)^{-1} \bP\big(\bx,-\tfrac{1}{2}-\eps)\\
&\,=\bR(\bx,-\eps)^{-1} \bP\big(\bx,-\mu)\\
&\,=\bR(\bx,-\eps)^{-1} \bs{\check{P}}(\bx,\mu)\,.
\esp\eeq
The TRBRs in this basis, where $\bC(\bx,\eps) = (2\pi i)^{-h}\mathds{1}$, take the form
\beq
\bH(\eps) = \bP(\bx,\mu)^T\bP(\bx,-\mu) = \bF_{\log}(\bx,\eps)^T\bR(\bx,-\eps)\bF_{\log}(\bx,-\eps)\,,
\eeq
which agrees with eq.~\eqref{quad2} upon identifying $\widetilde{\bH}(\eps)$ with $\bH(-\eps)$ and replacing $\eps$ by $-\eps$. Note that, while the quadratic relations were derived in the specific basis in which the intersection matrix is unity, we can now rotate the basis to any other basis, and obtain quadratic relations in that basis. We note, however, that the identification \eqref{eq:R_intersection} is specific to the basis where the intersection matrix is unity.

\paragraph{Conclusion.} To summarise, we see that all the conjectures and quadratic relations for maximal cuts of ref.~\cite{Lee:2018jsw} have a very natural interpretation and derivation in the context of twisted cohomology theories. In particular, conjectures about the special form of the differential equations for maximal cuts follow easily and naturally from the existence of logarithmic bases, and the appearance of the dimension-shift matrix in the quadratic relations is a consequence of the fact that it can be interpreted as a matrix of intersection numbers, cf.~eq.~\eqref{eq:R_intersection}.

\section{Examples}
\label{sec:examples}
\subsection{The unequal-mass sunrise integral}
\label{subsec:UneqSunrise}

As a first example, we consider the unequal-mass sunrise integral family in $D=2-2\varepsilon$ dimensions defined in eq.~\eqref{sunexstart}.
Quadratic relations satisfied by the maximal cuts of the equal-mass sunrise have been investigated in ref.~\cite{Lee:2018jsw}. We now argue that these quadratic relations can be extended to the unequal-mass case. One of the main differences is that, in the unequal-mass case there are four master integrals for the maximal cuts, while there are only two in the equal-mass case.

                             Using the
loop-by-loop approach, we obtain the Baikov polynomial in one internal variable $z$, and we define the twist 
\begin{align}
    \Phi_{{\tikz[baseline=-0.5ex]{\draw (0,0) circle [radius=2pt];
                             \draw (-4pt,0) -- (4pt,0);}}}(z) =
z^\varepsilon \left[(z-\lambda_1)(z-\lambda_2)
(z-\lambda_3)(z-\lambda_4)\right]^{
-\frac{1}{2}-\varepsilon}\, ,
\end{align}
where
\beq\bsp
 \lambda_1&\,=-(m_1-m_2)^2\,,\\
 \lambda_2&\,=-(m_1+m_2)^2\,,\\
 \lambda_3&\,=-(m_3-\sqrt{p^2})^2\,,\\
  \lambda_4&\,=-(m_3+\sqrt{p^2})^2\, .
\esp\eeq
This twist defines a twisted cohomology group $H_{\text{dR}}^n (X, \nabla_{\!\tikz[baseline=-0.5ex]{\draw (0,0) circle [radius=2pt]; \draw (-4pt,0) -- (4pt,0);}})$ with $X=\mathbb{C}-\{0,\lambda_1,\lambda_2,\lambda_3,\lambda_4
\}$. Note that in the limit $\varepsilon\rightarrow 0$, the Baikov polynomial reduces to the polynomial defining the elliptic curve attached to the sunrise integral, defined as the locus in $\mathbb{C}^2$ defined by
\begin{align}\label{curvesunrise}
y^2=p_{\tikz[baseline=-0.5ex]{\draw (0,0) circle [radius=2pt]; \draw (-4pt,0) -- (4pt,0);}}(z)= (z-\lambda_1)(z-\lambda_2)(z-\lambda_3)(z-\lambda_4)\, . 
\end{align}

Let us now discuss our choice of master integrals, or equivalently our choice of basis for $H_{\text{dR}}^n (X, \nabla_{\!\tikz[baseline=-0.5ex]{\draw (0,0) circle [radius=2pt]; \draw (-4pt,0) -- (4pt,0);}})$. A possible choice would be a logarithmic basis (and its regularised version), which leads to a particularly simple intersection matrix by Theorem~\ref{alphaformtheorem}.  
However, our basis choice is inspired from a geometrical viewpoint for the elliptic curve defined by eq.~\eqref{curvesunrise}, where the standard basis consists of Abelian differentials of the first, third and second kind. This will allow us to easily obtain analytic expressions for some of the maximal cuts in terms of complete elliptic integrals of first and second kind.
Explicitly, we choose the following basis for $H_{\text{dR}}^n (X, \nabla_{\!\tikz[baseline=-0.5ex]{\draw (0,0) circle [radius=2pt]; \draw (-4pt,0) -- (4pt,0);}})$:
\beq\bsp
\label{basissun}
\varphi_1^{{\tikz[baseline=-0.5ex]{\draw (0,0) circle [radius=2pt];\draw (-4pt,0) -- (4pt,0);}}}&\,=\rd z
\,,\\
 \varphi_2^{{\tikz[baseline=-0.5ex]{\draw (0,0) circle [radius=2pt];\draw (-4pt,0) -- (4pt,0);}}}&\,=\left(z^2 - \frac{s_1}{2} z + \frac{s_2}{6}\right)\rd z\,,\\
   \varphi_3^{{\tikz[baseline=-0.5ex]{\draw (0,0) circle [radius=2pt];\draw (-4pt,0) -- (4pt,0);}}}&\,=
z \,\rd z\,,\\
\varphi_4^{{\tikz[baseline=-0.5ex]{\draw (0,0) circle [radius=2pt];\draw (-4pt,0) -- (4pt,0);}}}&\,=\frac{\rd z}{z}\, ,
\esp\eeq
with $s_i$ elementary symmetric polynomials in the branch points $\lambda_1,\dots ,\lambda_4$. We pick a basis of dual co-cycles according to eq.~\eqref{eq:FI_dual_basis}:
\begin{align}
\check{\varphi}_i^{{\tikz[baseline=-0.5ex]{\draw (0,0) circle [radius=2pt];\draw (-4pt,0) -- (4pt,0);}}}=\left[\frac{\varphi_i^{{\tikz[baseline=-0.5ex]{\draw (0,0) circle [radius=2pt];\draw (-4pt,0) -- (4pt,0);}}}}{p_{\tikz[baseline=-0.5ex]{\draw (0,0) circle [radius=2pt]; \draw (-4pt,0) -- (4pt,0);}}(z)}\right]_c
\,,\qquad i=1,\ldots,4\,.
\end{align}

Next, let us discuss our choice of maximal cut contours, or equivalently our choice of basis for the twisted homology group and its dual.
For the dual homology basis we choose cycles supported on
\beq\bsp
\label{dualbas}
    \check{\gamma}_1^{{\tikz[baseline=-0.5ex]{\draw (0,0) circle [radius=2pt];\draw (-4pt,0) -- (4pt,0);}}}&\,=[\lambda_1,\lambda_2] \,,\\ 
    \check{\gamma}_2^{{\tikz[baseline=-0.5ex]{\draw (0,0) circle [radius=2pt];\draw (-4pt,0) -- (4pt,0);}}}&\,=[\lambda_2,\lambda_3]  \,,\\ 
    \check{\gamma}_3^{{\tikz[baseline=-0.5ex]{\draw (0,0) circle [radius=2pt];\draw (-4pt,0) -- (4pt,0);}}}&\,= [\lambda_1,\lambda_2]+[\lambda_3,\lambda_4]\,,\\
    \check{\gamma}_4^{{\tikz[baseline=-0.5ex]{\draw (0,0) circle [radius=2pt];\draw (-4pt,0) -- (4pt,0);}}}&\,= [\lambda_4,0]\, . 
\esp\eeq
This basis is depicted in figure~\ref{fig.cyclesell}, together with the canonical cycles $a$ and $b$ of the elliptic curve defined by eq.~\eqref{curvesunrise}. The basis of the twisted homology group is then supported on the regularised version of the dual cycles (cf.~eq.~\eqref{eq:check_to_h}):
\beq
\gamma_i^{{\tikz[baseline=-0.5ex]{\draw (0,0) circle [radius=2pt];\draw (-4pt,0) -- (4pt,0);}}} = \big[\check{\gamma}_i^{{\tikz[baseline=-0.5ex]{\draw (0,0) circle [radius=2pt];\draw (-4pt,0) -- (4pt,0);}}}\big]_{\textrm{reg}}\,,\qquad i =1\,\ldots,4\,.
\eeq

\begin{figure}[!t]
    \centering
    \includegraphics[align=c, scale=.4]{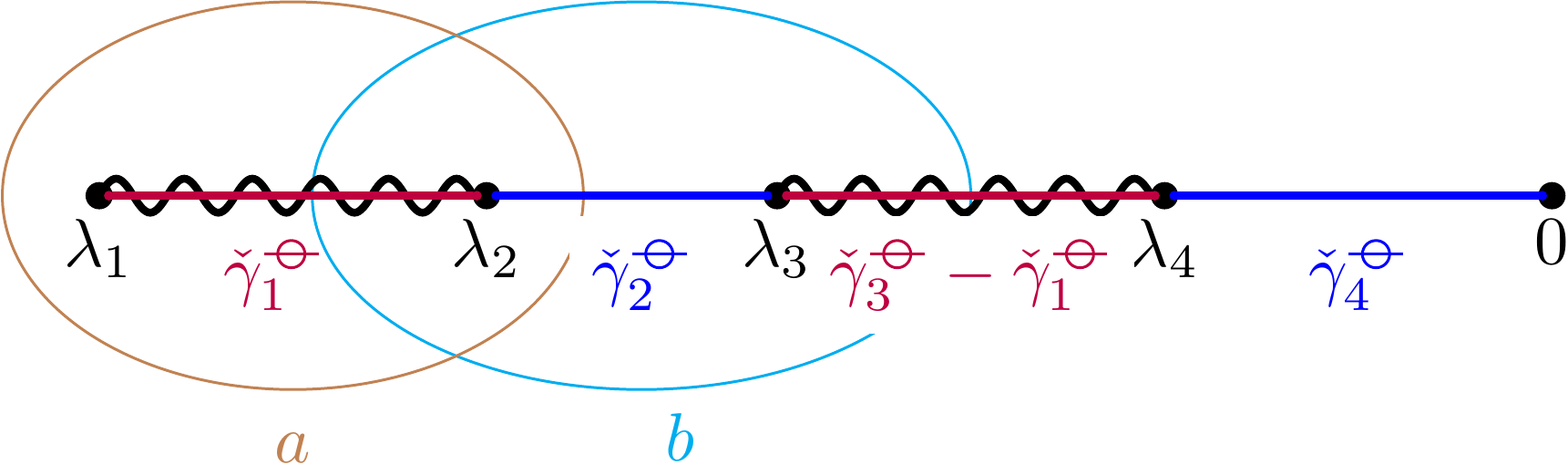}
 \caption{Our choice of (dual) twisted cycles is inspired by the geometric picture for the elliptic curve defined by eq.~\eqref{curvesunrise}. The two canonical cycles on the elliptic curve, $a$ and $b$, inspire our choice of the dual twisted cycles $\check{\gamma}_1$ and $\check{\gamma}_2$. In the $\varepsilon \rightarrow 0$ limit, the cycle $\check{\gamma}_3$ corresponds to the cycle around infinity, so that all differentials without residue integrate to zero, simplifying the period matrix. For generic $\varepsilon$ the twist has an additional factor $z^\varepsilon$ which requires us to add the dual twisted cycle $\check{\gamma}_4$.}
    \label{fig.cyclesell}
\end{figure}

With these choices for the bases of the (co-)homology groups and their duals, we obtain the following expressions for the intersection matrices (with $\bx = (p^2,m_1^2,m_2^2,m_3^2)$):
\begin{align} \label{csun}
\bC_{{\tikz[baseline=-0.5ex]{\draw (0,0) circle [radius=2pt];\draw (-4pt,0) -- (4pt,0);}}}(\bx,\eps)&= \left(
\begin{array}{cccc}
 0 & \frac{1}{-3 \varepsilon -1} & 0 & 0 \\
 \frac{1}{1-3 \varepsilon } & -\frac{\varepsilon  (\lambda_2+\lambda_2+\lambda_3+\lambda_4)^2}{12 \left(9 \varepsilon ^2-1\right)} & \frac{\lambda_1+\lambda_2+\lambda_3+\lambda_4}{6-18 \varepsilon } & 0 \\
 0 & -\frac{\lambda_1+\lambda_2+\lambda_3+\lambda_4}{18 \varepsilon +6} & -\frac{1}{3 \varepsilon } & 0 \\
 0 & 0 & 0 & -\frac{1}{\lambda_1 \lambda_2 \lambda_3 \lambda_4 \varepsilon } \\
\end{array}
\right)\,,\\
\bH_{{\tikz[baseline=-0.5ex]{\draw (0,0) circle [radius=2pt];\draw (-4pt,0) -- (4pt,0);}}}(\eps)&=\left(
\begin{array}{cccc}
 i \tan (\pi  \varepsilon ) & -1+\frac{1}{1+e^{2 i \pi  \varepsilon }} & i \tan (\pi  \varepsilon ) & 0 \\
 \frac{1}{1+e^{2 i \pi  \varepsilon }} & i \tan (\pi  \varepsilon ) & -i \tan (\pi  \varepsilon ) & 0 \\
 i \tan (\pi  \varepsilon ) & -i \tan (\pi  \varepsilon ) & 2 i \tan (\pi  \varepsilon ) & -1+\frac{1}{1+e^{2 i \pi  \varepsilon }} \\
 0 & 0 & \frac{1}{1+e^{2 i \pi  \varepsilon }} & i \csc (2 \pi  \varepsilon ) \\
\end{array}
\right)\,.
\end{align}
The period matrix $\bP_{{\tikz[baseline=-0.5ex]{\draw (0,0) circle [radius=2pt];\draw (-4pt,0) -- (4pt,0);}}}(\bx,\eps)$ can be expressed in terms of Lauricella functions, 
\beq\bsp\label{eq:lauricella}
    F_D^{(n)} (a,b_1,\ldots,&b_n;c;x_1,\ldots,x_n) =\\
    =&\, \frac{\Gamma(c)}{\Gamma(a)\Gamma(c-a)} \int_0^1 \rd t\, t^{a-1} (1-t)^{c-a-1}\prod_{i=1}^n(1-x_it)^{-b_i},
\esp\eeq
valid for $\mathrm{Re}(c)>\mathrm{Re}(a)>0$,
and the dual period matrix is given by 
\beq
\check{\bP}_{{\tikz[baseline=-0.5ex]{\draw (0,0) circle [radius=2pt];\draw (-4pt,0) -- (4pt,0);}}}(\bx,\eps) = \bP_{{\tikz[baseline=-0.5ex]{\draw (0,0) circle [radius=2pt];\draw (-4pt,0) -- (4pt,0);}}}(\bx,-\eps)\,.
\eeq

The TRBRs take the form:
\beq\label{eq:sunrise_trbr}
\frac{1}{2\pi i}\bP_{{\tikz[baseline=-0.5ex]{\draw (0,0) circle [radius=2pt];\draw (-4pt,0) -- (4pt,0);}}}(\bx,\eps)^T\big(\bC_{{\tikz[baseline=-0.5ex]{\draw (0,0) circle [radius=2pt];\draw (-4pt,0) -- (4pt,0);}}}(\bx,\eps)^{-1}\big)^T\bP_{{\tikz[baseline=-0.5ex]{\draw (0,0) circle [radius=2pt];\draw (-4pt,0) -- (4pt,0);}}}(\bx,-\eps) = \bH_{{\tikz[baseline=-0.5ex]{\draw (0,0) circle [radius=2pt];\draw (-4pt,0) -- (4pt,0);}}}(\eps)\,.
\eeq
We have checked numerically at various points that eq.~\eqref{eq:sunrise_trbr} holds. Since the entries of the period matrix are multi-valued functions, care is needed when evaluating the period matrix. 
We fix $\lambda_1<\lambda_2<\lambda_3<\lambda_4<0$ in the following. We note that the integral defining $P_{{\tikz[baseline=-0.5ex]{\draw (0,0) circle [radius=2pt];\draw (-4pt,0) -- (4pt,0);}},44}(\bx,\eps)$ is only defined for $\varepsilon>0$, while $\check{P}_{{\tikz[baseline=-0.5ex]{\draw (0,0) circle [radius=2pt];\draw (-4pt,0) -- (4pt,0);}},44}(\bx,\eps)$ is defined for $\varepsilon <0$, due to the occurrence of a non-integrable singularity at $z=0$ otherwise. In the TRBRs, we have to evaluate the expressions at a single value for $\varepsilon$. In order to circumvent this issue, we have performed a careful analytic continuation of the Lauricella function to the entire complex plane. After this analytic continuation, we find that that eq.~\eqref{eq:sunrise_trbr} holds numerically.

It is interesting to investigate eq.~\eqref{eq:sunrise_trbr} for $\epsilon=0$. Indeed, we expect that in the limit $\eps\to0$, four entries in the period matrix reduce to the periods and quasi-periods of the elliptic curve defined by eq.~\eqref{curvesunrise}. The latter are related by the well-known Legendre relation between elliptic integrals of the first and second kind, 
\begin{align}\label{eq:Legendre}
  \K(1-\lambda) \E(\lambda)+ \E(1-\lambda)\K(\lambda)-\K(\lambda)\K(1-\lambda)=\frac{\pi}{2}\,,
\end{align}
with 
\beq
\K(\lambda) = \int_{0}^1\frac{\rd t}{\sqrt{(1-t^2)(1-\lambda t^2)}}\qquad \textrm{~~and~~}
\E(\lambda) = \int_{0}^1\rd t\,\sqrt{ \frac{1-\lambda t^2}{1-t^2}}\,.
\eeq
Since the period matrix is a $4\times4$ matrix, it is interesting to see how the Legendre relation is recovered, and if there are any new relations at leading order in $\eps$ that go beyond it. We start by expanding eq.~\eqref{eq:sunrise_trbr} in $\eps$, and we observe that $H_{{\tikz[baseline=-0.5ex]{\draw (0,0) circle [radius=2pt];\draw (-4pt,0) -- (4pt,0);}},44}(\eps) = \mathcal{O}\left({\varepsilon^{-1}}\right)$. In order to obtain a finite limit, we introduce an additional rotation of the homology basis:
\begin{align}
    \bT_{{\tikz[baseline=-0.5ex]{\draw (0,0) circle [radius=2pt];\draw (-4pt,0) -- (4pt,0);}}}(\varepsilon)=\left(
\begin{smallmatrix}
 1& 0& 0& 0 \\
0 & 1 & 0 & 0 \\
 0 & 0 &1 & 0 \\
 0 & 0 &0 & \varepsilon \\
\end{smallmatrix}
\right)\,,
\end{align}
and we define 
\beq\bsp
\widetilde{\bP}_{{\tikz[baseline=-0.5ex]{\draw (0,0) circle [radius=2pt];\draw (-4pt,0) -- (4pt,0);}}}(\bx,
\eps)&\,=\bP_{{\tikz[baseline=-0.5ex]{\draw (0,0) circle [radius=2pt];\draw (-4pt,0) -- (4pt,0);}}}(\bx,\eps)\bT_{{\tikz[baseline=-0.5ex]{\draw (0,0) circle [radius=2pt];\draw (-4pt,0) -- (4pt,0);}}}(\varepsilon)\,,\\
\widetilde{\bH}_{{\tikz[baseline=-0.5ex]{\draw (0,0) circle [radius=2pt];\draw (-4pt,0) -- (4pt,0);}}}&\,=\bT_{{\tikz[baseline=-0.5ex]{\draw (0,0) circle [radius=2pt];\draw (-4pt,0) -- (4pt,0);}}}(\varepsilon)\bH_{{\tikz[baseline=-0.5ex]{\draw (0,0) circle [radius=2pt];\draw (-4pt,0) -- (4pt,0);}}}(\eps)\bT_{{\tikz[baseline=-0.5ex]{\draw (0,0) circle [radius=2pt];\draw (-4pt,0) -- (4pt,0);}}}(-\varepsilon)\,.
\esp\eeq
Equation~\eqref{eq:sunrise_trbr} then takes the form:
\beq\label{eq:sunrise_trbr_rotated}
\frac{1}{2\pi i}\widetilde{\bP}_{{\tikz[baseline=-0.5ex]{\draw (0,0) circle [radius=2pt];\draw (-4pt,0) -- (4pt,0);}}}(\bx,\eps)^T\big(\bC_{{\tikz[baseline=-0.5ex]{\draw (0,0) circle [radius=2pt];\draw (-4pt,0) -- (4pt,0);}}}(\bx,\eps)^{-1}\big)^T\widetilde{\bP}_{{\tikz[baseline=-0.5ex]{\draw (0,0) circle [radius=2pt];\draw (-4pt,0) -- (4pt,0);}}}(\bx,-\eps) = \widetilde{\bH}_{{\tikz[baseline=-0.5ex]{\draw (0,0) circle [radius=2pt];\draw (-4pt,0) -- (4pt,0);}}}(\eps)\,.
\eeq
Equation~\eqref{eq:sunrise_trbr_rotated} is now finite at $\eps=0$, and to leading order in $\eps$, we find
\begin{align}\label{eq:sunrise_trbr_0}
\frac{1}{2\pi i}\,(\eta_2 \omega_1-\eta_1 \omega_2)\, \left(
\begin{smallmatrix}
 0 & -1& 0 & 0 \\
 1 & 0 & 0 & 0 \\
 0 & 0 & 0 & 0 \\
 0 & 0 & 0 & 0 \\
\end{smallmatrix}
\right)=\left(
\begin{smallmatrix}
 0 & -\frac{1}{2} & 0 & 0 \\
\frac{1}{2} & 0 & 0 & 0 \\
 0 & 0 & 0 & 0 \\
 0 & 0 & 0 & 0 \\
\end{smallmatrix}
\right)\,,
\end{align}
where the left-hand side involves the periods and quasi-periods of the elliptic curve:
\beq\bsp
    \omega_1&\,=\frac{1}{c_4} \,\K(\lambda)\,,\\
     \omega_2&\,=\frac{i}{c_4}\,\K(1-\lambda)\,,\\
    \eta_1&\,=-2 c_4 \left[\E(\lambda)-\frac{2-\lambda}{3}\K(\lambda)\right]\,, \\
     \eta_2&\,=2i c_4 \,\left[\E(1-\lambda)+ \frac{1+\lambda}{3}\K(1-\lambda)\right]\,,
\esp\eeq
with $\lambda=\frac{(\lambda_1-\lambda_4) (\lambda_2-\lambda_3)}{(\lambda_1-\lambda_3) (\lambda_2-\lambda_4)}$
and $c_4=\frac{1}{2} \sqrt{(\lambda_1-\lambda_3) (\lambda_2-\lambda_4)}$. 
We see that eq.~\eqref{eq:sunrise_trbr_0} is equivalent to $\eta_2 \omega_1-\eta_1 \omega_2 = i\pi$, which is itself equivalent o the Legendre relation in eq.~\eqref{eq:Legendre}. Hence, to leading order in $\eps$, the TRBRs in eq.~\eqref{eq:sunrise_trbr_rotated} are equivalent to the Legendre relation.
%

\subsection{The non-planar crossed box}
\label{subsec:npBox}

\begin{figure}
\centering
\begin{tikzpicture}[scale=1.4]
    \draw[black, line width = 0.6mm] (-0.7,0.7) -- (1,0.7);
    \draw[black, line width = 0.6mm] (-0.7,-0.7) -- (1,-0.7);
    \draw[black, line width = 0.6mm] (-0.7,-0.7) -- (-0.7,0.7);
    \draw[black, line width = 0.6mm] (0.3,0) -- (1,0.7);
    \draw[black, line width = 0.6mm] (0.3,0) -- (1,-0.7);
    \draw[black, line width = 0.6mm] (1.7,0) -- (1,0.7);
    \draw[black, line width = 0.6mm] (1.7,0) -- (1,-0.7);

    \draw[black, thick] (-0.7,0.7) -- (-1,1);
    \draw[black, thick] (-0.7,-0.7) -- (-1,-1);
    \draw[black, thick] (0.3,0) -- (0,0);
    \draw[black, thick] (1.7,0) -- (2,0);

    \draw[black, thick] (-0.9,-1.2) node {$p_1$};
    \draw[black, thick] (-0.9,1.2) node {$p_2$};
    \draw[black, thick] (2.2,0) node {$p_3$};
    \draw[black, thick] (0,0.2) node {$p_4$};
\end{tikzpicture}
\caption{Non-planar crossed box diagram. All internal propagators are massive.}
\label{fig:npcb}
\end{figure}
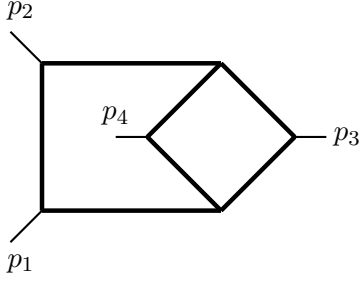

As a second example we will consider the non-planar crossed box in $D=4-2\varepsilon$ dimensions, see figure \ref{fig:npcb}:
\begin{equation}
    I_{\bm{\nu}}^{\mathrm{npcb}}(s,t,m^2)=e^{2\gamma_{\mathrm{E}}\varepsilon}\int\frac{\mathrm{d}^D\ell_1}{i\pi^{D/2}}\int\frac{\mathrm{d}^D\ell_2}{i\pi^{D/2}}\frac{1}{D_1^{\nu_1}\cdots D_7^{\nu_7}}\,,
\end{equation}
with the propagators
\beq\bsp
    &D_1=\ell_1^2-m^2,\quad D_2=(\ell_1-p_1)^2-m^2,\quad D_3=(\ell_1-p_1-p_2)^2-m^2  \\
    &D_4=\ell_2^2-m^2,\quad D_5=(\ell_2-p_3)^2-m^2,\quad D_6=(\ell_1+\ell_2)^2-m^2 \\
    &D_7=(\ell_1+\ell_2-p_1-p_2-p_3)^2-m^2\, ,
\esp\eeq
and all external momenta massless $p_i^2=0$. The integral only depends on two dimensionless invariants which we take to be the Mandelstam variables
\begin{equation}
    s=(p_1+p_2)^2\,,\quad t=(p_2+p_3)^2\,, \qquad u=(p_1+p_3)^2 = -s-t\,.
\end{equation}
In the following we set $m^2=1$ without loss of generality.

This integral has been studied in refs.~\cite{Huang:2013kh,Georgoudis:2015hca,Marzucca:2023gto}, and was shown to be connected to a hyperelliptic curve of genus two. Using the loop-by-loop Baikov approach we can show that the maximal cuts of this integral are twisted periods with the twist
\begin{equation}
    \Phi_{\mathrm{npcb}}(z)=\left[(z-\lambda_1)(z-\lambda_2)\right]^{-\frac{1}{2}}\left[(z-\lambda_3)(z-\lambda_4)(z-\lambda_5)(z-\lambda_6)\right]^{-\frac{1}{2}-\varepsilon}\,,
\end{equation}
with the roots given by
\beq\bsp
\label{oneformshyper}
    &\lambda_1=-\frac{1}{4}\left(s+\sqrt{s(s-4)}\right)\,,\quad \lambda_2=-\frac{1}{4}\left(s-\sqrt{s(s-4)}\right)\,,  \\
    &\lambda_3=-\frac{1}{4}\left(s+\sqrt{s(s+12)} \right)\,,\quad \lambda_4=-\frac{1}{4}\left(s-\sqrt{s(s+12)} \right)\,, \\
    &\lambda_5=-\frac{s(s+t)+2\sqrt{s^2 t+s t^2}}{2s}\,, \quad \lambda_6=-\frac{s(s+t)-2\sqrt{s^2 t+s t^2}}{2s}\,. 
\esp\eeq
For $\varepsilon=0$ the twist reduces to the square root of a polynomial of degree six which defines a hyperelliptic curve of genus two, see, e.g., ref.~\cite{farkasKra} for more mathematical details on Riemann surfaces (of higher genus). The twisted cohomology group $H_{\mathrm{dR}}^1(X,\nabla_{\mathrm{npcb}})$ with $X=\mathbb{C}-\{\lambda_1,\dots, \lambda_6\}$ is five-dimensional. We will choose the following basis of (classes of) twisted co-cycles
\beq\bsp
    &\varphi_1^{\mathrm{npcb}}=\mathrm{d}z,\quad \varphi_2^{\mathrm{npcb}}=z\,\mathrm{d}z,\quad \varphi_3^{\mathrm{npcb}}=-\left(z^4-\frac{3}{4}s_1z^3+\frac{s_2}{2}z^2-\frac{s_3}{4}z \right)\mathrm{d}z \\
    &\varphi_4^{\mathrm{npcb}}=-\frac{1}{2}\left(z^3-\frac{s_1}{2}z^2 \right)\mathrm{d}z,\quad \varphi_5^{\mathrm{npcb}}=z^2\mathrm{d}z,
\esp\eeq
where the $s_k$ are elementary symmetric polynomials in the branch points $\lambda_1,\dots \lambda_6$. As before we will choose the dual twisted cohomology basis such that the dual period matrix can be obtained from the original period matrix by changing the sign of $\varepsilon$:
\begin{equation}
    \check{\varphi}_i^{\mathrm{npcb}}=\left[\frac{1}{p_{\mathrm{npcb}}(z)}\varphi_i^{\mathrm{npcb}}\right]_c,\quad i=1,\dots,6\,,
\end{equation}
where $y^2=p_{\mathrm{npcb}}(z)=(z-\lambda_1)\cdots(z-\lambda_6)$ is the defining equation for the hyperelliptic curve $\mathcal{C}$ associated to the non-planar crossed box.

The above choice of basis for the twisted cohomology group is particularly natural from a geometric point of view, motivated by the appearance of the hyperelliptic curve. The one-forms on any Riemann surface can be grouped into three types: differentials of the first kind which are holomorphic on the entire surface, differentials of the second kind which are meromorphic but have no residues and differentials of the third kind which are meromorphic with non-zero residues. The one-forms $\left. \Phi_{\mathrm{npcb}} \right|_{\varepsilon=0}\varphi_i^{\mathrm{npcb}}$ of eq.~\eqref{oneformshyper} constitute  a basis of differentials of the first kind ($i=1,2)$, a basis of differentials of the second kind ($i=3,4$) and the differential of the third kind with a pole at infinity ($i=5$) on a hyperelliptic curve of genus two. Let us stress that, although the basis choice is motivated by the $\varepsilon=0$ geometry, it  still constitutes a valid basis for $\varepsilon\neq 0$.

Our choice for the homology basis is also be inspired by the hyperelliptic geometry. A hyperelliptic curve of genus two admits a canonical homology basis of dimension four with the basis elements conventionally labelled $a_1,a_2,b_1,b_2$, see figure \ref{fighyper}.
\begin{figure}[!t]
    \centering
    \includegraphics[align=c, scale=.3]{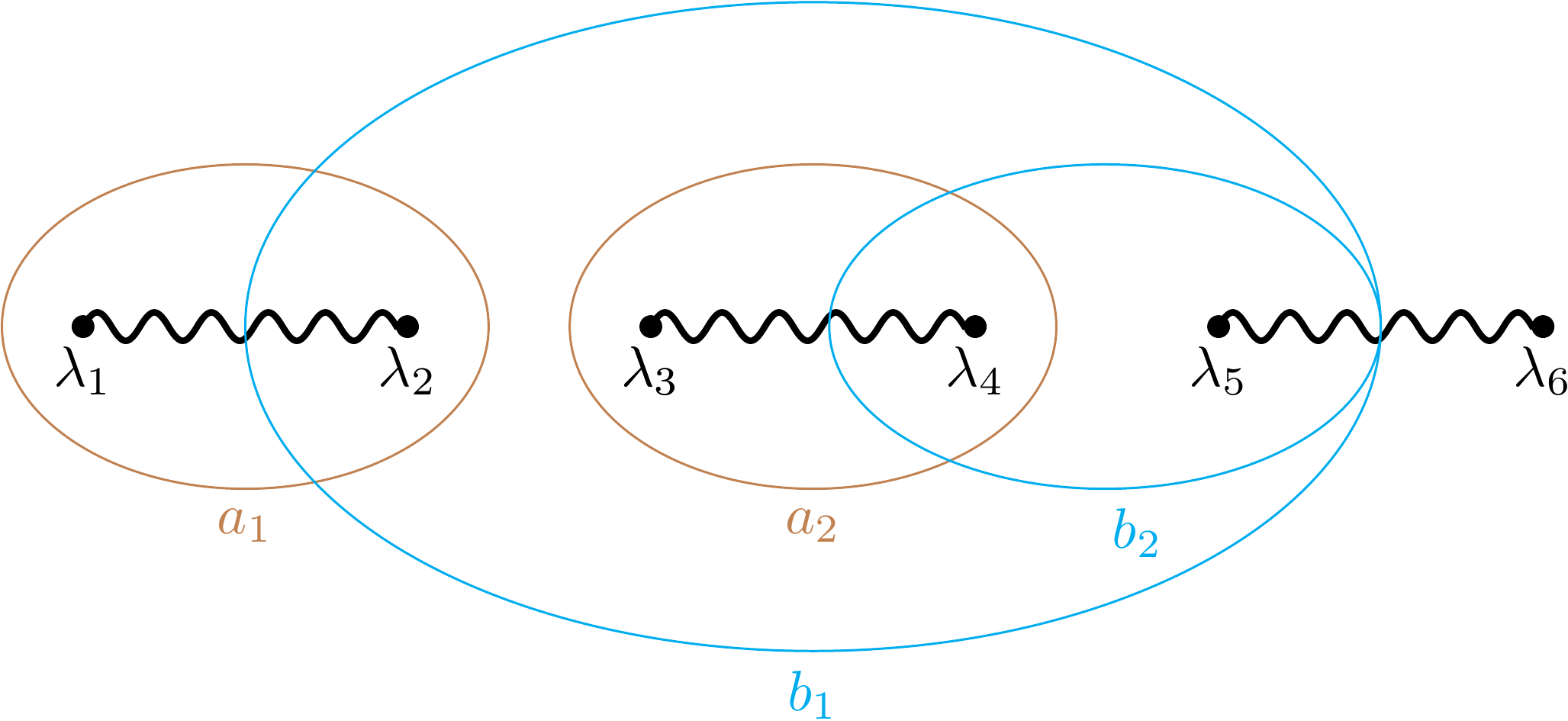}
 \caption{A canonical choice of cycles for a hyperellitpic curve of genus two with branch points $\lambda_1,\dots,\lambda_6$.}
    \label{fighyper}
\end{figure}
Since crossing the branch cut simply changes the sign of the one-forms chosen above (for $\varepsilon=0$) one can reduce integrations around these cycles to simple integrations along (sums of)  intervals. This leads us to our choice of twisted dual homology basis
\beq\bsp
    &\check{\gamma}_1^{\textrm{npcb}}=[\lambda_1,\lambda_2]\,,\quad \check{\gamma}_2^{\textrm{npcb}}=[\lambda_3,\lambda_4]\,, \\
    &\check{\gamma}_3^{\textrm{npcb}}=[\lambda_2,\lambda_3]+[\lambda_4,\lambda_5]\,,\quad \check{\gamma}_4^{\textrm{npcb}}=[\lambda_4,\lambda_5]\,, \\
    &\check{\gamma}_5^{\textrm{npcb}}=[\lambda_1,\lambda_2]+[\lambda_3,\lambda_4]+[\lambda_5,\lambda_6]\,,
\esp\eeq
with the basis of the twisted homology group being made up of the regularised versions of the respective dual basis elements. Note that for $\eps=0$ $\check{\gamma}_1^{\textrm{npcb}},\check{\gamma}_2^{\textrm{npcb}}$ correspond to the $a$-cycles and $\check{\gamma}_3^{\textrm{npcb}},\check{\gamma}_4^{\textrm{npcb}}$ correspond to the $b$-cycles on the hyperelliptic curve. The final basis element $\check{\gamma}_5^{\textrm{npcb}}$ is needed to complete the basis for $\varepsilon\neq 0$ and is chosen in such a way that the differential one-forms $ \Phi_{\mathrm{npcb}} \varphi_i^{\mathrm{npcb}}$ with $i=1,\dots,4$ integrate to zero for $\varepsilon=0$, leading to a  block structure in the period matrix.

We can now compute the (co-)homology intersection matrices $\widetilde{\bm{C}}_{\mathrm{npcb}}(\bx,\eps)$ and $\widetilde{\bm{H}}_{\mathrm{npcb}}(\eps)$ (with $\bx = (s,t,m^2)$), as well as the period matrix $\widetilde{\bm{P}}_{\mathrm{npcb}}(\bx,\eps)$. The dual period matrix is then by construction obtained by just changing the sign of $\varepsilon$.
The period matrix entries can all be expressed in terms of four-variable Lauricella functions from eq.~\eqref{eq:lauricella} (with $n=4$).

We already discussed above that the dual cycle $\check{\gamma}_5^{\textrm{npcb}}$ is not independent for $\varepsilon=0$, we will now similarly decouple the differential of the third kind by performing a slight redefinition of the cohomology basis, namely by rescaling the co-cycle $\varphi_5^{\mathrm{npcb}}$ with $\varepsilon$. In the new basis, we have 
\beq\bsp\label{eq:npcb-rotation}
  \bm{C}_{\mathrm{npcb}}(\bx,\eps)&\,=\bm{B}({\varepsilon})\widetilde{\bm{C}}_{\mathrm{npcb}}(\bx,\eps)\bm{B}({-\varepsilon}) \,,\\
    \bm{P}_{\mathrm{npcb}}(\bx,\eps)&\,=\bm{B}({\varepsilon})\widetilde{\bm{P}}_{\mathrm{npcb}}(\bx,\eps)\,,
\esp\eeq
with $\bm{B}({\varepsilon})=\mathrm{diag}(1,1,1,1,\varepsilon)$. The rescaled intersection matrices can be found in an ancillary file made available together with the arXiv submission of this paper. After this additional rescaling, the TRBRs take the form
\begin{equation}
    \bm{P}_{\mathrm{npcb}}(\bx,\eps)\left( \bm{H}_{\mathrm{npcb}}(\eps)^{-1}\right)^T{\bm{P}}^T_{\mathrm{npcb}}(\bx,-\eps)=2\pi i\,\bm{C}_{\mathrm{npcb}}(\bx,\eps)\,.
\end{equation}
Just like for the sunrise integral, the TRBRs diverge in the limit. If we expand around $\eps=0$, the contributions of order $\ord(\eps^{-1})$ only give trivial relations, while at $\ord(\eps^0)$, we find
\begin{equation}
    \begin{pmatrix}
        \bm{\mathcal{B}}\bm{\mathcal{A}}^T-\bm{\mathcal{A}}\bm{\mathcal{B}}^T & 
        \bm{\mathcal{B}}\bm{\mathcal{\tilde{A}}}^T-\bm{\mathcal{A}}\bm{\mathcal{\tilde{B}}}^T & \bm{0} \\
        \bm{\mathcal{\tilde{B}}}\bm{\mathcal{A}}^T-\bm{\mathcal{\tilde{A}}}\bm{\mathcal{B}}^T &
        \bm{\mathcal{\tilde{B}}}\bm{\mathcal{\tilde{A}}}^T-\bm{\mathcal{\tilde{A}}}\bm{\mathcal{\tilde{B}}}^T & \bm{0} \\
        \bm{0} & \bm{0} & 0
    \end{pmatrix}= 
    \begin{pmatrix}
        \bm{0} & -i\pi \mathds{1} & \bm{0} \\
        i\pi \mathds{1} & \bm{0} & \bm{0} \\
        \bm{0} & \bm{0} & 0
    \end{pmatrix},
\end{equation}
where $\bm{\mathcal{A}},\bm{\mathcal{B}},\bm{\mathcal{\tilde{A}}},\bm{\mathcal{\tilde{B}}}$ are the $2\times 2$ block matrices appearing in the period matrix for $\varepsilon=0$
\begin{equation}
    \left.\bm{P}_{\mathrm{npcb}}\right(\bx,{\varepsilon=0})= \begin{pmatrix}
        \bm{\mathcal{A}} & \bm{\mathcal{B}} & \bm{0} \\
        \bm{\mathcal{\tilde{A}}} & \bm{\mathcal{\tilde{B}}} & \bm{0} \\
        \bm{0} & \bm{0} & 0
    \end{pmatrix}\,.
\end{equation}
The matrices $\bm{\mathcal{A}},\bm{\mathcal{B}}$ are obtained by integrating the two basis differentials of the first kind over the two $a$-cycles and $b$-cycles respectively and are usually referred to as $a$- and $b$-periods. Similarly the matrices $\bm{\mathcal{\tilde{A}}},\bm{\mathcal{\tilde{B}}}$ are obtained by integrating the two basis differentials of the second kind over the two $a$-cycles and $b$-cycles respectively and are called $a$- and $b$-quasi-periods. 

As already anticipated in section \ref{subsec:quadraticsCY}, the quadratic relations between the $2\times 2$ block matrices $\bm{\mathcal{A}},\bm{\mathcal{B}},\bm{\mathcal{\tilde{A}}},\bm{\mathcal{\tilde{B}}}$, which can be seen as a generalization of the Legendre relation, follow directly from the classical Riemann bilinear relations on the hyperelliptic curve $\mathcal{C}$ (cf., e.g., ref.~\cite{farkasKra})
\begin{equation}
    \sum_{i=1}^2\left[\int_{a_i}\omega\int_{b_i}\eta- \int_{b_i}\omega\int_{a_i}\eta\right]=\int_{\partial \mathcal{C}}f \eta\,.
\end{equation}
where $\omega,\eta$ are one-forms and $f$ is a primitive of $\omega$, i.e.~$\mathrm{d}f=\omega$, see e.g., ref.~\cite{2012arXiv1208.0990B} for details. The integral on the right-hand side is over the boundary of the fundamental domain (obtained by cutting the surface along all homology cycles) and can be seen as the sum of residues on the hyperelliptic curve. 
Note that we could also expand the other form of the TRBR
\begin{equation}
    \bm{P}_{\mathrm{npcb}}(\bx,\eps)^T\left( \bm{C}_{\mathrm{npcb}}(\bx,\eps)^{-1}\right)^T{\bm{P}}_{\mathrm{npcb}}(\bx,-\eps)=2\pi i\,\bm{H}_{\mathrm{npcb}}(\eps)
\end{equation}
in $\epsilon$. This would yield an alternative form of the generalized Legendre relations, for example
\begin{equation}
    \bm{\mathcal{\tilde{A}}}^T\bm{\mathcal{B}}-\bm{\mathcal{A}}^T\bm{\mathcal{\tilde{B}}}=-i\pi\mathds{1}\,.
\end{equation}

\section{Conclusion}
\label{sec:concl}

In this paper we have initiated the first systematic study of quadratic relations between (cut) Feynman integrals in dimensional regularisation. We work within the mathematical framework of twisted cohomology, and our starting point are the TRBRs in eq.~\eqref{generalriemann}, which relate the intersection matrices for the twisted cycles and co-cycles and the period matrix and its dual. We show that for a given family of Feynman integrals, it is possible to define a period matrix whose entries are cuts of the master integrals. The dual period matrix, however, is not directly expressible in terms of cuts, preventing the interpretation of the TRBRs as quadratic relations for cut Feynman integrals. The origin of this can be traced back to the fact that, due to the appearance of propagator poles not `regulated' by the twist, one needs to work with relative twisted cohomology, where the twisted co-cycles and their duals are not related in a straightforward manner. We stress that we do not claim that there are no TRBRs for non-maximally cut Feynman integrals! The TRBRs exist, but they are separately linear in the period matrix and its dual. Since the latter are not expected to be related in a simple manner (and our discussion in example~\ref{example21} shows that indeed they are not), the TRBRs can in general not be interpreted as quadratic relations among entries of the period matrix, i.e., they can in general not be interpreted as quadratic relations among cut integrals.

We then move on to study maximal cuts, which can be studied using a non-relative twisted cohomology theory. This allows us to relate the twisted (co-)cycles and their duals in a simple manner. As a consequence, when working with these bases, the period matrix for maximal cuts and its dual are related by changing the sign of the dimensional regulator $\eps$, and the TRBRs can then be interpreted as quadratic relations for maximal cuts. We show how these quadratic relations reduce to quadratic relations among maximal cuts in full kinematics that have appeared in the literature before\footnote{We do not include the quadratic relations from ref.~\cite{Broadhurst:2018tey} into our comparison, because they hold for banana integrals in special kinematics.}, and we find that the latter are a special case of the TRBRs. We also present two new results for quadratic relations for the maximal cuts of the unequal-mass sunrise integral and the non-planar crossed box, and we show that to leading order in $\eps$ we recover the well known Legendre relations and classical Riemann bilinear relations for periods of Riemann surfaces.

Finally, we stress an important point: while our analysis shows that it is in general not possible to obtain quadratic relations among non-maximally cut integrals from TRBRs, we cannot exclude that quadratic relations for non-maximal cuts exist. Indeed, we have only shown that the TRBRs lead to quadratic relations whenever one can work with a non-relative twisted cohomology theory, where the period matrix and its dual are related by changing the sign of $\eps$. The question whether there are other non-linear relations, possibly of even higher polynomial degree, that cannot be obtained as a consequence of the well-known linear relations and the TRBRs is a fascinating open question that may deserve further investigation in the future.

\section*{Acknowledgements}
CS and FP thank Andrzej Pokraka for helpful discussions and comments. The work of CS is supported by the CRC 1639 ``NuMeriQS'', and the work of CD and FP is funded by the European Union
(ERC Consolidator Grant LoCoMotive 101043686). Views
and opinions expressed are however those of the author(s)
only and do not necessarily reflect those of the European
Union or the European Research Council. Neither the
European Union nor the granting authority can be held
responsible for them.

\begin{appendix}
\section{TRBRs and deformations}
\label{app:deform}

In this appendix we discuss the TRBRs for the hypergeometric ${}_2F_1$ function from example~\ref{example21}.  Since the factor $(1-x^{-1}z)$ appears with a vanishing exponent in the twist in eq.~\eqref{eq_reltwist}, relative twisted cohomology is the appropriate framework to study the period and intersection matrices and the TRBRs. The goal of this appendix is to show that we could also start from the deformed twist
\begin{align}
    \Phi({\rho})= z^{\alpha_0 }(1-z)^{\alpha_1 }\left(1-x^{-1}z\right)^{\rho}\,,\qquad \rho\neq0\,,
\end{align}
and still recover the same TRBRs as in the relative case.

For $\rho\neq0$, the condition~\eqref{restrict} is satisfied, and we can use the results from example~\ref{example11} for the intersection and period matrices. If we work in  the basis of eq.~\eqref{herenotchange} (with $\alpha_x=\rho$),  the Laurent expansion of the intersection matrices in eqs.~\eqref{cmat} around $\rho=0$ is
\beq\bsp
\label{355}
\bs{C}(x,\balpha,\rho)&=\left(\begin{smallmatrix} -\frac{1}{\rho}+\mathcal{O}(\rho^{0})&0\\ 0&-\frac{\alpha_1}{\alpha_0 \alpha_{01}}\end{smallmatrix}\right)\,,\\
\bs{H}(\balpha,\rho)&=\tfrac{1}{2\pi i}\left(\begin{smallmatrix} \frac{1}{\rho}+\mathcal{O}(\rho^{0})&0\\ 0&\pi\cot(\pi \alpha_0) +\pi\cot(\pi \alpha_1)+\mathcal{O}(\rho)\end{smallmatrix} \right)\, . 
\esp\eeq
We see that different entries scale differently with $\rho$. In order to extract the leading behaviour in the limit $\rho\to0$, we first rescale the different entries, which is equivalent to performing a rotation with the matrix
\begin{align}
\label{traf}
     \bs{T}(\rho)&=\left(\begin{smallmatrix} 
\rho &0 \\ 0 &1
    \end{smallmatrix}\right)\, . 
\end{align}
We then have
\beq\bsp
    \bs{C}(x,\balpha,\rho) \bs{T}(\rho)&\,= \bs{C}^{\text{rel}}(x,\balpha) + \ord(\rho)\,,\\ 
    \bs{T}(\rho)\bs{H}(\balpha,\rho) &\,=\bs{H}^{\text{rel}}(\balpha)+\ord(\rho)\, , 
\esp\eeq
where $\bs{C}^{\text{rel}}$ and $\bs{H}^{\text{rel}}$ are the intersection matrices obtained from the relative framework in the case $\rho=0$ given in eq.~\eqref{intersectionsrelative}, in agreement with the results from ref.~\cite{Brunello:2023fef}. 
We can apply the same strategy to the period matrix and its dual, and we see that we recover the results from the relative case in eqs.~\eqref{pmatrel} and~\eqref{pmatreldual} in the limit $\rho\to0$ for $x>1$:
\beq\bsp
\bs{P}(x,\balpha,\rho)\bs{T}(\rho) &\,= \bs{P}^{\text{rel}}(x,\balpha)+\ord(\rho) \,,\\
\bs{T}(\rho) \bs{\check{P}}(x,\balpha,\rho)&\,= \bs{\check{P}}^{\text{rel}}(x,\balpha) + \ord(\rho)\, . 
\esp\eeq
We know that 
\beq
\label{eq:deformed_condition}
\bs{\check{P}}(x,\balpha,\rho)=\bs{{P}}(x,-\balpha,-\rho)\,,\qquad \rho \neq0
\eeq
but the rotation needed to extract the leading term violates this relation, which is reflected in the fact that in the relative case we have $\bs{\check{P}}^{\text{rel}}(x,\balpha)\neq\bs{{P}}^{\text{rel}}(x,-\balpha)$.

It is interesting to ask if there is any quadratic relation of the form
\beq\label{eq:P_rel_sym}
\bs{{P}}^{\text{rel}}(x,\balpha)^T\bs{R}(\bx,\eps)\bs{{P}}^{\text{rel}}(x,-\balpha) = \bs{Q}(\eps)\,.
\eeq
To understand this point, we define the rotation of the deformed period matrix and this dual in a way that preserves the relation~\eqref{eq:deformed_condition} even for $\rho=0$:
\beq\bsp
\bs{P}(x,\balpha,\rho)\bs{T}(\rho) &\,=  \bs{P}^{\text{rel}}(x,\balpha)+\ord(\rho) \,,\\
\bs{\check{P}}(x,\balpha,\rho)\bs{T}(-\rho)&\,= \bs{P}(x,-\balpha,-\rho)\bs{T}(-\rho) = \bs{{P}}^{\text{rel}}(x,-\balpha) + \ord(\rho)\, . 
\esp\eeq
Then it is easy to check that eq.~\eqref{eq:P_rel_sym} is satisfied with
\beq\bsp
\bs{R}(x,\eps) &\,= \lim_{\rho\to0}\left(\bC(\rho)^{-1}\right)^T= \left(\begin{smallmatrix}0&0\\0&-\frac{\alpha_1}{\alpha_0 \alpha_{01}}\end{smallmatrix}\right)+ \ord(\rho)\,,\\
\bs{Q}(\eps) &\,= \bs{T}(\rho)\bH(\eps)\bs{T}(-\rho) = \tfrac{1}{2\pi i}\left(\begin{smallmatrix} 0&0\\0&\pi\cot(\pi\alpha_0)+\pi\cot(\pi\alpha_1) \end{smallmatrix}\right) + \ord(\rho)\,.
\esp\eeq
Thus, we see that there is a quadratic relation of the form~\eqref{eq:P_rel_sym}.
However, it cannot arise from TRBRs, because the matrices $\bs{R}(x,\eps)$ and $\bs{Q}(\eps)$ do not have full rank (while intersection matrices computed for a basis always have full rank).

\section{On the $\varepsilon$-dependence of specific intersection numbers}
\label{app:proof}

The goal of this appendix is to show that the $\mu$-dependence of the intersection numbers involving a derivative $\rdex$ is fully captured by eq.~\eqref{eq:eta_phi_scale}.

We start by reviewing the theorem of ref.~\cite{ojm_1200788347} for the intersection numbers of $\rd\!\log$-forms, which leads to eq.~\eqref{eq:Matsumoto_thm}. The proof of eq.~\eqref{eq:eta_phi_scale} then follows the same lines as the proof in ref.~\cite{ojm_1200788347}. Following ref.~~\cite{ojm_1200788347}, we assume that the twist is given by 
\beq
\Phi=\prod_{i=0}^{r+1}L_i(\bs{z},\bs{x})^{\alpha_i}\,,
\eeq
where the $L_i(\bs{z},\bs{x})=0$ define hyperplanes $\mathcal{L}_i$. The proof, however does not crucially rely on this assumption.
%
%
%
We assume that we can find a basis of the cohomology group that consists of $\rd\!\log$ form as in eq.~\eqref{logbasis}.
We introduce the following symbol for index sets $P_m=(q_1,\dots,q_m), P_{m+1}=(p_0,\dots,p_m)$ with $p_0<p_1<\dots<p_m$ and similarly for the $q_i$
\begin{align}
    \delta(P_m, P_{m+1}) &=\begin{cases} (-1)^\mu\,, &\text{ if } P_m\subset P_{m+1}\,,\\ 0\,, & \text{ if } P_m \not\subset P_{m+1}\,,\end{cases},
\end{align}
where $\mu$ is determined as $\{p_\mu\}=P_{m+1}-P_m$ in the case $P_m \subset P_{m+1}$.
The following theorem is proven in ref.~\cite{ojm_1200788347}:
\begin{theorem}
    \label{matsumototheorem}
    For multi-indices 
    \begin{align*}
        I=(i_0,\dots, i_n )\,, &\text{ with } 0\leq i_0< i_1<\dots<i_n\leq r+1\,,\\
        J=(j_0,\dots, j_n )\,, &\text{ with } 0\leq j_0< j_1<\dots<j_n\leq r+1\,,
    \end{align*}
the intersection pairing of the corresponding $\rd\! \log$-forms is 
\begin{align}
    \langle \varphi_J|\check{\varphi}_I\rangle =\begin{cases}
        (2\pi i)^n\frac{\sum_{i\in I}\alpha_i}{\prod_{i\in I}\alpha_i}\,, &\text{ if } I=J\,,\\
        (2\pi i)^n\frac{(-1)^{\beta_1+\beta_2}}{\prod_{i\in (I\cap J)} \alpha_i}\,, &\text{ if } |I\cap J| = n\,, \\
        0\,,& \text{ otherwise} \, . 
    \end{cases}
\end{align}
The exponents $\beta_i$ are defined by $\{i_{\beta_1}\}=I-(I\cap J) $ and  $\{i_{\beta_2}\}=J-(I\cap J) $.
\end{theorem}
Note that for the choice of $\alpha_i= a_i\mu$ as in eq.~\eqref{formtwist}, this implies eqs.~\eqref{eq:Matsumoto_thm} and~\eqref{eq:Matsumoto_thm_2}.
The proof of the theorem relies on an explicit construction of the compactly supported forms and the computation of residues. In a similar manner, we will prove the following theorem: 
\begin{theorem}
We consider a $\rd\!\log$-form $\varphi_I$ and a covariant derivative of a $\rd \!\log$-form $\varphi_J$, namely $\eta_J=\left(\rd_{\mathrm{ext}} +\rd_{\mathrm{ext}} \log \Phi \wedge\right) \varphi_J$ defined by the multi-indices $I,J$. Then, the intersection pairing between these differentials exhibits the following $\epsilon$-dependence:
\begin{align}
    \langle \eta_J |\varphi_I\rangle \sim \frac{1}{\mu^{n-1}}\, . 
\end{align}
\end{theorem}
\begin{proof}
We follow closely the proof of ref.~\cite{ojm_1200788347} and consider first the case $n=1$. 
    \paragraph{The case $n=1$.} 
Since $\varphi_I$ is a $\rd \!\log$-form, all steps of the proof of Theorem~\ref{matsumototheorem} that only affect this form still apply. In that way, we can use the construction of its compactly supported version given in ref.~\cite{ojm_1200788347}. We can then directly go to the next step of the proof, from where on our proof diverges (slightly). We start by writing
\begin{align}
\label{middlestep}
   \langle \eta_J |\varphi_I\rangle = -2\pi i \sum_{p=0}^{r+1} \text{Res}_{\mathcal{L}_p} \left(\psi^p \eta_J \right)\, ,
\end{align}
where $\psi^p$ is a holomorphic primitive of $\varphi_I$ that has the following expansion in local coordinates near $\mathcal{L}_p$ \cite{ojm_1200788347}:
\begin{align}
    \psi^p=\frac{\delta(p,I)}{\alpha_p} +\mathcal{O}(z)\, .  
\end{align}
To understand better the behaviour of the differential form $\eta_J$ in local coordinates near the $\mathcal{L}_p$, we first give the explicit expression for it. We have:
\begin{align}
\label{c1form}
    \rd_{\text{ext}} \varphi_J &= \frac{\rd_{\text{ext}}(\rd L_{j_0})}{L_{j_0}}-\frac{\rd_{\text{ext}}(\rd L_{j_1})}{L_{j_1}}-\frac{\left(\rd_{\text{ext}}L_{j_0}\right)\wedge\left(\rd L_{j_0}\right)}{L_{j_0}^2}+\frac{\left(\rd_{\text{ext}}L_{j_1}\right)\wedge\left(\rd L_{j_1}\right)}{L_{j_1}^2}\,,\\
\label{c2form}
    \rd_{\text{ext}} \log\Phi\wedge \varphi_J&= \sum_{l=0}^{r+1} \frac{\alpha_l \rd_{\text{ext}}L_l}{L_l} \wedge \rd \log \frac{L_{j_0}}{L_{j_1}}\nonumber\\
    &=\Bigg(\frac{\alpha_{j_0} (\rd_{\text{ext}}L_{j_0})\wedge (\rd L_{j_0})}{L_{j_0}^2}+\sum_{l\neq j_0} \frac{\alpha_l \rd_{\text{ext}}L_l}{L_l} \wedge \frac{\rd L_{j_0}}{L_{j_0}}\Bigg)-(j_0\leftrightarrow j_1)\nonumber \, . 
\end{align}
Note that the exterior differential $\rd=\rd z\,\partial_z$ is to be considered with respect to the internal variable $z$. If we inserted this form into eq.~\eqref{middlestep} we could not immediately deduce the scaling with $\mu$, because we still have double poles as well as terms that do not scale with $\mu$ in eq.~\eqref{c1form}. 
This is because the residue is not either 1 or zero as it is the case for logarithmic forms. Of course, if we go through with the full residue calculation, we get the same scaling in the end, but we use a strategy, where we can read off the scaling of the intersection pairing easily. Namely, we add  to $\eta_J$ an exact form to write $\eta_J$ in a simpler form which the scaling can be read off. Consider the following exact form 
\begin{align}
    \chi_{j_0}^{1} =& \nabla \left(\frac{\rd_{\text{ext}} L_{j_0}}{L_{j_0}} - \frac{\rd_{\text{ext}} L_{j_1}}{L_{j_1}} \right)\\
     =& \Bigg(\frac{-\rd_{\text{ext}}\left( \rd L_{j_0}\right)}{L_{j_0}} +\frac{\left(\rd_{\text{ext}} L_{j_0}\right)\wedge\left(\rd L_{j_0}\right)}{L_{j_0}^2}-\frac{\alpha_{j_0} (\rd_{\text{ext}} L_{j_0})\wedge (\rd L_{j_0})}{L_{j_0}^2}\notag \\&+\sum_{l\neq j_0}\frac{\alpha_l \rd L_l}{L_l} \wedge \frac{\rd_{\text{ext}}L_{j_0}}{L_{j_0}}
    \Bigg)  -(j_0\leftrightarrow j_1)\,.\notag 
\end{align}
We than have:
\begin{align}\label{sumofsimpl}
    \eta_J&\sim\eta_J+\chi_{j_0}^1 =:\tilde{\eta}_J= \sum_{l\neq j_0} \alpha_l\left(\frac{ \rd_{\text{ext}}L_l}{L_l} \wedge \frac{\rd L_{j_0}}{L_{j_0}}+\frac{\rd L_l}{L_l} \wedge \frac{\rd_{\text{ext}}L_{j_0}}{L_{j_0}}\right)-(j_0\leftrightarrow j_1)\,.
\end{align}
All terms in the sum in eq.~\eqref{sumofsimpl} have only simple poles at each of the hyperplanes $\mathcal{L}_j$, and letting $\alpha_j= a_j\mu$ the $\mu$-dependence of the sum is given by an overall $\mu$. Thus in local coordinates $z^p$ near $\mathcal{L}_p$, we obtain: 
\begin{align*}
    \psi^p \eta_J= \frac{\delta_{p,I}}{\mu a_p} \frac{\mu a_p C_p \,\rd z^p}{z^p} + \mathcal{O}(1) \, , 
\end{align*} 
where $C_p$ is an $\mu$-independent factor that can be determined from eq.~\eqref{sumofsimpl}. Finally, we obtain: 
\begin{align}
\label{finstep1}
   \langle \eta_J |\varphi_I\rangle = -2\pi i \sum_{p=0}^{r+1} \text{Res}_{\mathcal{L}_p} \left(\psi^p \eta_J \right)\sim \mu^0 
\end{align}
is independent of $\mu$. 
\paragraph{The case of generic $n$.} We now sketch how the previous argument can be generalised to arbitrary $n$.
Again, we can start with a residue calculation, as all previous steps can be obtained via the same compactification of the $\rd\!\log$ form $\varphi_J$ as in ref.~\cite{ojm_1200788347}: 
\begin{align}
    \langle \eta_J|\varphi_I\rangle = (-2\pi i)^n \sum_P \text{Res}_{z_n} \left( \text{Res}_{z_{n-1}}\left(\dots \text{Res}_{z_1} \left( \tilde{\psi}^P \eta_J\right)\dots\right)\right)\, . 
\end{align}
Here, $z_i^P$ denote the local coordinates near $\mathcal{L}_P$. $\tilde{\psi}^P$ is a holomorphic primitive with leading order term proportional to $\prod_{i\in P}\frac{1}{\alpha_i}$. We start again by writing out the form $\eta_J$: 
\begin{align}
    \eta_J&= \sum_{l=0}^{n-1}(-1)^{l}  \rd\!\log \left(\frac{L_{j_0}}{L_{j_1}}\right)\wedge \dots \wedge \rdex \rd\! \log \left(\frac{L_{j_{l}}}{L_{j_{l+1}}}\right)\wedge \dots \wedge \rd\! \log \left(\frac{L_{j_{n-1}}}{L_{j_n}}\right) \\
    &\, \, \, \, \, \,  + \sum_{l=0}^{r+1} \alpha_l \rdex\! \log L_l \wedge \varphi_J\, . \notag
\end{align}
Additionally, we define the exact form
\begin{align}
    \chi_{j_{l-1}}^n &= \nabla \left(\varphi_{j_0j_1} \wedge \dots \wedge \varphi_{j_{l-2},j_{l-1}}\wedge \left(\rdex\!\log L_{j_{l-1}} -\rdex\! \log L_{j_l}\right) \wedge \varphi_{j_l,j_{l+1}} \wedge\dots \wedge \varphi_{j_{n-1},n}\right)\, . 
\end{align}
where we recall that $\varphi_{j_{k-1}j_{k}}$ are defined according to \eqref{logbasis}.
Then 
\begin{align}
    \eta_J\sim \eta_J +\sum_{l=1}^{n} \chi_{j_{l-1}}^n =:\tilde{\eta}_J 
\end{align}
has only simple poles in each set of local coordinates $\bs{z}^P$ and scales with $\mu$ overall. The terms with double poles cancel, as in the $n=1$ case. Thus: 
\begin{align}
    \langle\tilde{\eta}_J|\varphi_I\rangle \sim \frac{\mu}{\mu^n} =\frac{1}{\mu^{n-1}} \, . 
\end{align}
\end{proof}

\end{appendix}

\bibliographystyle{JHEP}
\bibliography{refs.bib}

\end{document}